\documentclass[submission,copyright,creativecommons]{eptcs}

\usepackage{array}
\usepackage[ruled,vlined,linesnumbered]{algorithm2e}
\usepackage[dvipsnames]{xcolor}
\usepackage{amsmath,amsfonts,amssymb,amsthm}
\usepackage[utf8]{inputenc}
\usepackage[T1]{fontenc}
\usepackage{graphicx}
\usepackage{tikz}
\usetikzlibrary{backgrounds,arrows,automata,positioning,calc,shapes.misc}
\usepackage{hyperref}
\usepackage[backgroundcolor=magenta!10!white]{todonotes}
\usepackage{setspace}
\usepackage{enumitem}
\usepackage{ifthen}
\usepackage{microtype}
\usepackage{stackrel}
\usepackage{proof}
\usepackage[space]{cite}
\usepackage{stmaryrd}
\usepackage{cleveref}
\usepackage{multirow}
\usepackage{subcaption}
\usepackage{verbatim}
\usepackage{placeins}
\usepackage{booktabs}
\usepackage{thmtools}

\usepackage{relsize}
\usepackage{multirow}
\usepackage{multicol}

\usepackage{mathtools}

\usepackage{pgfplots}
\usepackage{standalone}
\usepgfplotslibrary{fillbetween}
\pgfplotsset{width=3.5cm, compat=1.8, grid style={dashed}}

\renewcommand{\P}{\mathcal{P}}

\newcommand{\prb}{\mathbf{Pr}}
\newcommand{\Ab}{\mathbf{A}}
\newcommand{\Ib}{\mathbf{I}}

\newcommand{\xb}{\mathbf{x}}

\newcommand{\bb}{\mathbf{b}}

\newcommand{\M}{\mathcal{M}}

\newcommand{\C}{\mathcal{C}}

\newcommand{\Act}{\operatorname{Act}}
\newcommand{\supp}{\operatorname{supp}}

\newcommand{\goal}{\operatorname{Goal}}
\newcommand{\fail}{\operatorname{fail}}

\newcommand{\Cyl}{\operatorname{Cyl}}
\renewcommand{\S}{\mathfrak{S}}
\renewcommand{\Pr}{\mathrm{Pr}}
\newcommand{\fin}{{\operatorname{fin}}}
\newcommand{\Paths}{\operatorname{Paths}}
\newcommand{\Post}{\operatorname{post}}
\newcommand{\Pre}{\operatorname{pre}}

\newcommand{\dtpw}{\textsf{dtpw}}
\newcommand{\dppw}{\textsf{dppw}}
\newcommand{\utw}{\textsf{utw}}
\newcommand{\utpw}{\textsf{utpw}}
\newcommand{\Dw}{\textsf{Dw}}
\newcommand{\dtw}{\textsf{dtw}}
\newcommand{\dtps}{\textsf{DTP}}

\newcommand{\partsubsysmap}{\texttt{psubsys}}

\newcommand{\minval}{\operatorname{min-val}}
\newcommand{\maxval}{\operatorname{max-val}}
\newcommand{\interface}{\operatorname{inc}}
\newcommand{\outint}{\operatorname{out}}

\newcommand{\pointval}{\operatorname{val}}

\newcommand{\pre}{\operatorname{parent}}
\newcommand{\sucs}{\operatorname{children}}
\newcommand{\reach}{\operatorname{reach}}
\newcommand{\dom}{\operatorname{dom}}
\newcommand{\closure}{\operatorname{cl}}
\newcommand{\out}{\operatorname{exit}}

\newcommand{\lst}[2]{\operatorname{\textbf{l}}_{#1}^{#2}}
\newcommand{\rst}[2]{\operatorname{\textbf{r}}_{#1}^{#2}}
\newcommand{\llay}[1]{\operatorname{left}_{#1}}
\newcommand{\rlay}[1]{\operatorname{right}_{#1}}
\newcommand{\lay}[1]{\operatorname{layer}_{#1}}

\newcommand{\gurobi}{\textsc{Gurobi}}
\newcommand{\cbc}{\textsc{Cbc}}

\newcounter{simoncomments}

\colorlet{simonColor}{Yellow!30!white}
\colorlet{florianColor}{Brown!30!white}
\colorlet{jakobColor}{Magenta!30!white}

\theoremstyle{plain}

\newtheorem{thm}{Theorem}
\newtheorem{theorem}[thm]{Theorem}

\newtheorem{lemma}[thm]{Lemma}

\newtheorem{definition}[thm]{Definition}

\theoremstyle{definition}

\newtheorem{remark}[thm]{Remark}

\title{Witnessing subsystems for probabilistic systems \\ with low tree width}
\author{Simon Jantsch \qquad Jakob Piribauer \qquad Christel Baier
  \institute{Technische Universität Dresden\thanks{This work was supported by DFG grant 389792660 as part of TRR~248 -- CPEC, see \url{https://perspicuous-computing.science}, the Cluster of Excellence EXC 2050/1 (CeTI, project ID 390696704, as part of Germany’s Excellence Strategy), DFG-projects BA-1679/11-1, BA-1679/12-1 and the Research Training Group QuantLA (GRK 1763).}}
  \email{\{simon.jantsch, jakob.piribauer, christel.baier\}@tu-dresden.de}}
\def\titlerunning{Witnessing subsystems for probabilistic systems with low tree width}
\def\authorrunning{Simon Jantsch, Jakob Piribauer and Christel Baier}

\begin{document}
\maketitle

\begin{abstract}
  A standard way of justifying that a certain probabilistic property holds in a system is to provide a witnessing subsystem (also called critical subsystem) for the property.
  Computing \emph{minimal} witnessing subsystems is NP-hard already for acyclic Markov chains, but can be done in polynomial time for Markov chains whose underlying graph is a tree.
  This paper considers the problem for probabilistic systems that are \emph{similar} to trees or paths.
  It introduces the parameters \emph{directed tree-partition width} (\dtpw) and \emph{directed path-partition width} (\dppw) and shows that computing minimal witnesses remains NP-hard for Markov chains with bounded \dppw{} (and hence also for Markov chains with bounded \dtpw{}).
  By observing that graphs of bounded \dtpw{} have bounded width with respect to all known tree similarity measures for directed graphs, the hardness result carries over to these other tree similarity measures.
  Technically, the reduction proceeds via the conceptually simpler \emph{matrix-pair chain problem}, which is introduced and shown to be NP-complete for nonnegative matrices of fixed dimension.
  Furthermore, an algorithm which aims to utilise a given directed tree partition of the system to compute a minimal witnessing subsystem is described.
  It enumerates partial subsystems for the blocks of the partition along the tree order, and keeps only necessary ones.
  A preliminary experimental analysis shows that it outperforms other approaches on certain benchmarks which have directed tree partitions of small width. 
\end{abstract}

\section{Introduction}

The ability to justify and explain a verification result is an important feature in the context of verification.
For example, classical model checking algorithms for linear temporal logic (LTL) return a \emph{counterexample} if the system does not satisfy the property.
In this case a counterexample is typically an ultimately periodic trace of the system which does not satisfy the formula.
\emph{What} constitutes a valid counterexample is dependent on both the kind of system and the kind of specification that is considered~\cite{ClarkeV2003}.

For \emph{probabilistic systems}, which can be modeled using \emph{discrete-time Markov chains} (DTMC) or the more general \emph{Markov decision processes} (MDP), a single trace of the system is usually not enough to justify that a given (probabilistic) property holds\cite{AbrahamBDJKW2014}.
A typical class of properties in the probabilistic setting are \emph{probabilistic reachability constraints}, where one asks whether the (maximal or minimal, for MDPs) probability to reach a set of goal states satisfies a threshold condition.
The motivation for considering reachability (apart from the fact that it is a fundamental property) is that one can treat $\omega$-regular properties using methods for reachability over the product of the system with an automaton for the property~\cite{WimmerJAKB2014}.
To justify that the probability to reach the goal is higher than some threshold in a DTMC one can return a \emph{set of traces} of the system whose probability exceeds the threshold\cite{HanK2007}.
Another notion of counterexamples for probabilistic reachability constraints are \emph{witnessing subsystems} (also called \emph{critical subsystems})~\cite{WimmerJABK2012,WimmerJAKB2014,FunkeJB2020}.
The idea is to justify a lower bound on the reachability probability by providing a \emph{subsystem} which \emph{by itself} already exceeds the threshold.
Apart from providing explanations to a human working with the system model, probabilistic explanations have been used in automated analysis frameworks such as \emph{counterexample-guided abstraction refinement}~\cite{HermannsWZ2008} or \emph{counterexample-guided inductive synthesis}~\cite{AndriushchenkoCJK2021a,CeskaHJK2019}.
In all these applications it is important to find \emph{small} explanations (all paths of the system or the entire system as subsystem are trivial explanations in case the property holds).

Computing \emph{minimal} witnessing subsystems in terms of number of states is computationally difficult.
The corresponding decision problem, henceforth called the \emph{witness problem}, is NP-complete already for acyclic DTMCs~\cite{FunkeJB2020}.
Known algorithms rely on \emph{mixed-integer linear programming} (MILP)~\cite{WimmerJABK2012,WimmerJAKB2014,FunkeJB2020} or \emph{vertex enumeration}~\cite{FunkeJB2020}.
%Known heuristics either enumerate probable paths and add them iteratively to the subsystem~\cite{JansenAKWKB2011} or try to find small solutions of certain linear inequation systems using linear-programming~\cite{FunkeJB2020}.
On the other hand, the problem is in P for DTMCs whose underlying graph is a tree~\cite{FunkeJB2020}.
This leads to the natural question of whether efficient algorithms exist for systems whose underlying graph is \emph{similar to} a tree.
A parameter that measures this is the \emph{tree width} of a graph, which has been studied extensively in graph theory~\cite{Bodlaender1997}.
Several NP-hard problems for graphs, for example the $3$-coloring problem, are in P for graphs with bounded tree width~\cite{Bodlaender1997}.
Faster algorithms for standard problems in probabilistic model checking were proposed for systems of small tree width~\cite{AsadiCGMP2020a,ChatterjeeL2013}.
Algorithms for non-probabilistic quantitative verification problems on models with low tree width were considered in~\cite{ChatterjeeIP2021}.
One motivation for studying such systems is that \emph{control-flow graphs} of programs in languages such as JAVA and C, under certain syntactic restrictions, are known to have bounded tree width~\cite{Thorup1998,GustedtMT2002}.
A stronger notion than tree width is \emph{path width}~\cite{RobertsonS1983}, which intuitively measures how similar a graph is to a path and has been applied in fields such as graph drawing~\cite{DujmovicFKLMNRRWW2008} and natural language processing\cite{KornaiT1992}.

The standard notion of tree width is defined for undirected graphs.
Related notions have been considered for directed graphs, although here the theory is not as mature and as of now there is no standard notion~\cite{JohnsonRST2001, Safari2005, Reed1999}.
A stronger notion than tree width for undirected graphs is that of \emph{tree-partition width}~\cite{Wood2009,Seese1985,Halin1991}, which requires a partition of the graph whose induced quotient structure is a tree.
To the best of our knowledge, \emph{tree-partition width} has not been studied for directed graphs so far.

\noindent \textbf{Contributions.}%
\begin{itemize}
\item The paper introduces a tree-similarity measure called directed tree-partition width (\dtpw), which can be seen as the directed analogue to the tree-partition width for undirected graphs that is known from the literature. We show that deciding whether there exists a directed tree partition with width at most $k$ is NP-complete (\Cref{sec:treelikesystems}).
\item 
The second main contribution is NP-completeness for the witness problem in  Markov chains with bounded \dtpw{} (\Cref{sec:hardnesswitness}).
This implies that the problem is NP-hard for bounded-width Markov chains with respect to the following tree-similarity measures for directed graphs: directed tree width (from\cite{JohnsonRST2001}), bounded D-width (from\cite{Safari2005}) and bounded undirected tree width.
The reduction proceeds via the conceptually simpler $d$-dimensional matrix-pair chain problem, which we introduce and show to be NP-complete (\Cref{sec:hardnessmcp}) for fixed $d$, even for nonnegative matrices.
\item We describe a dedicated algorithm that computes minimal witnesses by proceeding bottom-up along a given directed tree partition.
  It enumerates partial subsystems for the blocks of the partition, but keeps only necessary ones (\Cref{sec:algorithm}).
  On certain instances that do have a good tree-decomposition, our prototype implementation significantly outperforms the standard MILP approach (\Cref{sec:experiments}).
\end{itemize}

\section{Preliminaries}

\noindent \textbf{Graphs, partitions and quotients.}
A set of subsets $\{S_1,\ldots,S_n\}$ of a given set $S$ is a \emph{partition} of $S$ if $S_i \cap S_j = \varnothing$ for all $1 \leq i < j \leq n$ and $\bigcup_{1 \leq i \leq n} S_i = S$.
The elements of a partition are called \emph{blocks}.
Given a graph $G = (V,E)$ and a partition $\mathcal{P} = \{V_1,\ldots,V_n\}$ of $V$, the \emph{quotient} of $G$ under $\mathcal{P}$ is the graph $(\mathcal{P},E_{\mathcal{P}})$, where $(V_i,V_j) \in E_{\mathcal{P}}$ iff there is some $(s,t) \in E$ such that $s \in V_i$ and $t \in V_j$.
A directed graph is a \emph{tree} if its underlying undirected graph is a tree and every node has in-degree at most one.
A directed graph is a \emph{path} if it is a tree and all vertices have outdegree at most one.

\noindent \textbf{Probabilistic systems.}
A \emph{discrete-time Markov chain} (DTMC) $\M$ is a tuple $(S,P,\iota)$ where $S$ is a set of states, $P : S \times S \to \mathbb{Q}_{\geq 0}$ is the \emph{probabilistic transition matrix}, which needs to satisfy $\sum_{s' \in S} P(s,s') \leq 1$ for all $s \in S$, and where the \emph{initial distribution} $\iota$
satisfies $\sum_{s\in S} \iota (s)\leq 1$.
\footnote{We require the transition matrix and initial distribution to be \emph{sub-stochastic} here. 
To obtain a \emph{stochastic} transition matrix and initial distribution one can add a state $\fail$ together with edges to $\fail$ carrying the missing probability. Using sub-stochastic DTMCs technically simplifies the treatment of subsystems and allows to ignore the $\fail$ state in the underlying graph.}
%% , we could add a state $\fail$ with $P(\fail,\fail)=1$ and add transitions with the missing probability mass from each state $s$ to $\fail$, i.e., $P(s,\fail)=1-\sum_{s' \in S} P(s,s')$. Similarly, the initial distribution would be extended by $\iota(\fail)=1-\sum_{s\in S} \iota (s)$.
%% The treatment of subsystems later on, however, is technically slightly simpler in the sub-stochastic formulation.}\todoJakob{check footnote}
A state $s$ with $\sum_{s' \in S} P(s,s')=0$ is called a trap state. A path is a finite or infinite sequence $s_0 s_1 s_2 \ldots \in S^{\leq\omega}$ such that $\iota(s_0) > 0$ and for all $i \geq 0: \; P(s_i,s_{i+1}) > 0$.
A \emph{maximal path} is  a path that is infinite or that ends in a trap state.
The set of maximal paths of $\M$ is called $\Paths(\M)$ and carries a (sub-)probability measure whose associated $\sigma$-algebra is generated by the cylinder sets $\Cyl(\tau) = \{\pi \in \Paths(\M) \mid \tau \text{ is a prefix of } \pi\}$, which have probability $\Cyl(s_0 \ldots s_n) = \iota(s_0) \cdot \prod_{0 \leq i < n} P(s_i,s_{i+1})$.
The probability of a measurable set $\Pi \subseteq \Paths(\M)$ is denoted by $\Pr_{\M}(\Pi)$.
A \emph{Markov decision process} (MDP) $\M$ is a tuple $(S,\Act,P,\iota)$ where $S$ is a set of states, $\Act : S \to 2^{A}$ is a function that assigns to each state a finite set of \emph{actions} from the set $A$, $P : S \times A \times S \to \mathbb{Q}_{\geq 0}$ satisfies $\sum_{s' \in S} P(s,\alpha,s') \leq 1$ for all $(s,\alpha)$ with  $s \in S$ and $\alpha \in \Act(s)$, and $\iota$ satisfies $\sum_{s\in S} \iota (s)\leq 1$.
The  paths  of $\M$ are finite or infinite sequences  $s_0 \alpha_0 s_1 \alpha_i \ldots \in (S \times A)^{\leq \omega} $ such that $\iota(s_0) > 0$ and for all $i \geq 0: \; P(s_i,\alpha_i,s_{i+1}) > 0$ and $\alpha_i \in \Act(s_i)$.
The set $\Paths_\fin(\M)$ denotes the finite paths ending in a state.
The set $\Paths(\M)$ is the set of maximal paths, which are infinite or end in a state $s$ with $\sum_{s^\prime\in S}P(s,\alpha,s^\prime)=0$ for all actions $\alpha\in \Act(s)$.
A \emph{scheduler} of $\M$ is a function $\S : \Paths_\fin(\M) \to A$ satisfying $\S(s_0 \alpha_0 \ldots s_n) \in \Act(s_n)$.
Every scheduler $\S$ of $\M$ induces a (possibly infinite) Markov chain $\M_\S$ and thereby a (sub-)probability measure on $\Paths(\M)$.
The maximal, respectively minimal, probabilities of some path property $\Pi \subseteq \Paths(\M)$ are defined as $\prb^{\max}_\M(\Pi) = \sup_{\S} \Pr_{\M_{\S}}(\Pi)$ and $\prb^{\min}_\M(\Pi) = \inf_{\S} \Pr_{\M_{\S}}(\Pi)$, where $\S$ ranges over all schedulers of $\M$.
For $\omega$-regular properties this notation is justified as the supremum, respectively infimum, is attained by some scheduler.
For more details see~\cite[Chapter 10]{BaierK2008}.
The underlying graph of an MDP $\M = (S,\Act,P,\iota)$ has vertices $S$ and edges: $\{(s,s') \in S\times S \mid \text{there exists } \alpha \in \Act \text{ such that } P(s,\alpha,s') > 0\}$.
We denote by $\M(s)$ the MDP one gets by replacing the initial distribution in $\M$ by the dirac distribution on $s \in S$.

\noindent \textbf{Witnessing subsystems.} Let $\M = (S,\Act,P,\iota)$ be an MDP.
A \emph{subsystem} of $\M$ is an MDP $\M' = (S',\Act,P',\iota')$ where $S' \subseteq S$ 
and for all $(s,\alpha)$ with $s \in S'$ and $\alpha \in \Act(s)$ and $s' \in S' : \ P'(s,\alpha,s') \in \{P(s,\alpha,s'), 0\}$.
Similarly, we require for all $s \in S : \ \iota'(s) \in \{\iota(s),0\}$.
In words, edges of the subsystem either retain the probability of the original system, or have probability zero. This again results in a sub-stochastic MDP.
%In the latter case the lost probability needs to be added to the edge to $\fail$ which is the only one that is not restricted in the above definition.\footnote{This step is not necessary in the formulation with sub-stochastic systems.}
The subsystem \emph{induced by} a set of states $S' \subseteq S$ is defined as $\M_{S'} = (S',\Act,P',\iota')$ where $P'(s,\alpha,s') = P(s,\alpha,s')$ if $s,s' \in S'$, and otherwise $P'(s,\alpha,s') = 0$, and similarly for $\iota'$.
Given an MDP $\M' = (S',\Act,P',\iota')$, we assume that there is a set of trap states $\goal$.
By $\lozenge \goal$, we denote the set of paths that contain a state $t\in \goal$. 
Any subsystem $\M'$ of $\M$ satisfies $\prb_{\M'}^{\max}(\lozenge \goal) \leq \prb_{\M}^{\max}(\lozenge \goal)$ and $\prb_{\M'}^{\min}(\lozenge \goal) \leq \prb_{\M}^{\min}(\lozenge \goal)$ as the probability to reach $\goal$ cannot increase under any scheduler by setting transition probabilities to $0$.
The subsystem $\M'$ is a \emph{witness} for $\prb^{*}_{\M}(\lozenge \goal) \geq \lambda$ if $\prb^{*}_{\M'}(\lozenge \goal) \geq \lambda$, where $* \in \{\min,\max\}$ and $\lambda \in [0,1]$.
The definition for DTMCs is analogous.

\section{Directed tree- and path-partition width}
\label{sec:treelikesystems}

We propose a natural extension of tree-partition width~\cite{Wood2009} to directed graphs.
In what follows, let $G = (V,E)$ be a fixed finite directed graph.

\begin{definition}[Directed tree partition]
  A partition $\mathcal{P} = \{V_1,\ldots,V_n\}$ of $V$ is a \emph{directed tree partition} of $G$ if the quotient of $G$ under $\mathcal{P}$ is a tree.
  We denote by $\dtps(G)$ the set of directed tree partitions of $G$.
\end{definition}

\begin{definition}[Directed tree-partition width (\dtpw)]
  The \emph{directed tree-partition width} of graph $G$ is:
  \[ \emph{\dtpw}(G) := \min_{\mathcal{P} \in \dtps(G)} \; \max_{S \in \mathcal{P}} \; |S|\]
\end{definition}
Replacing \emph{tree} by \emph{path} in definitions 1 and 2 leads to the notions of \emph{directed path partition} and \emph{directed path-partition width} (\dppw).
Any strongly connected component of a graph needs to be included in a single block of the partition, which distinguishes this notion from other notions of tree width for directed graphs.
In particular, it is not the case that the standard tree width of an undirected graph $G_u$ equals the directed tree-partition width of the directed graph one gets by including both edges $(u,v)$ and $(v,u)$ whenever $u$ and $v$ are connected in $G_u$.

\begin{restatable}{proposition}{relationToOtherNotions}
  If a class $\C$ of graphs has bounded directed tree-partition width, then $\C$ has bounded directed tree width (from\cite{JohnsonRST2001}), bounded D-width (from\cite{Safari2005}) and bounded undirected tree width.
\end{restatable}

Deciding whether a directed tree partition with small width exists is NP-complete.
The reduction goes from the \emph{oneway bisection problem}~\cite{FeigeY2003} in directed graphs, which asks whether there exists a partition of the given graph into two equally-sized vertex sets $V_0,V_1$ such that all edges go from $V_0$ to $V_1$.

\begin{restatable}{proposition}{treePartitionWithNPHardness}
  The two problems \emph{(1)} decide $\emph{\dppw(G)} \leq k$ and \emph{(2)} decide $\emph{\dtpw(G)} \leq k$, given a directed graph $G$ and $k \in \mathbb{N}$, are both NP-complete.
\end{restatable}

\section{Hardness of the witness problem for DTMCs with low directed path-partition width}
\label{sec:hardnesswitness}
The witness problem for probabilistic reachability in DTMCs is defined as follows.
\begin{definition}[Witness problem]
  The \emph{witness problem} takes as input a DTMC $\M$, $k \in \mathbb{N}$ and $\lambda \in \mathbb{Q}$, and asks whether there exists a witnessing subsystem for $\Pr_{\M}(\lozenge \goal) \geq \lambda$ with at most $k$ states.
\end{definition}
The problem is NP-complete for acyclic DTMCs, but in P for DTMCs whose underlying graph is a tree\cite{FunkeJB2020}.
In this section we prove NP-hardness of the witness problem for DTMCs with directed path-partition width at most 6 (and hence also for DTMCs with directed tree-partition at most 6). The idea is to relate different subsets of the blocks of a directed path partition to matrices that describe how the reachability probability is passed on through the system.
This leads us to a natural intermediate problem:

\begin{definition}[$d$-dimensional matrix-pair chain problem]
  The $d$-dimensional matrix-pair chain problem ($d$-MCP) takes as input a sequence $(M_0^1,M_1^1),\ldots,(M_0^n,M_1^n)$, where $M_i^j \in \mathbb{Q}^{d \times d}$, a starting vector $\iota \in \mathbb{Q}^{1 \times d}$, final vector $f \in \mathbb{Q}^{d \times 1}$, and $\lambda \in \mathbb{Q}$ (with all numbers encoded in binary) and asks whether there exists a tuple $(\sigma_1,\ldots, \sigma_n) \in \{0,1\}^n$ such that:
  \[ \iota \cdot M_{\sigma_1}^1 \cdots M_{\sigma_n}^n \cdot f \geq \lambda \]
  The \emph{nonnegative} variant of the problem restricts all input numbers to be nonnegative.
\end{definition}
In the following we will show that the $2$-MCP is NP-hard and can be reduced to the \emph{nonnegative} $3$-MCP.
The reason that we are interested in the nonnegative variant is that we would like to reduce it to the witness problem, where one can naturally encode (sub)stochastic matrices, but it is unclear how to deal with negative values.
\newpage
\noindent The chain of polynomial reductions that the argument uses is the following:
\begin{figure}[h!]
  \centering
  \resizebox{.7\linewidth}{!}{
  \begin{tikzpicture}
    \node[draw,rounded rectangle,scale=1,align=center,minimum size=1.5cm] (mcp) {$2$-MCP};
    \node[draw,rounded rectangle,scale=1,align=center,minimum size=1.5cm] (nnmcp) [right= 2cm of mcp] {nonnegative $3$-MCP};
    \node[draw,rounded rectangle,scale=1,align=center,minimum size=1.5cm] (mwp) [right= 2cm of nnmcp] {witness problem \\ $\dppw{} = 6$};
  
    \draw[->,thick,shorten <= 0.1cm,shorten >= 0.1cm] (mcp.east) -- (nnmcp.west);
    \draw[->,thick,shorten <= 0.1cm,shorten >= 0.1cm] (nnmcp.east) -- (mwp.west);
\end{tikzpicture}}
\end{figure}
\subsection{Hardness of the matrix-pair chain problem}
\label{sec:hardnessmcp}

To show NP-hardness of the $2$-MCP we reduce from the \emph{partition} problem, which is among Karp's 21 NP-complete problems\cite{Karp1972}.
Given a finite set $S = \{s_1, \ldots, s_n\} \subseteq \mathbb{Z}$, whose elements are encoded in binary, it asks whether there exists $W \subseteq S$ such that $\sum W = \sum (S \setminus W)$, where $\sum X = \sum_{x \in X} x$.
The main idea of the polynomial reduction is to relate each $s_i$ to a pair of matrices $M^{i}_0,M^{i}_1$ where $M^{i}_0$ realizes a clockwise rotation by an angle which corresponds to the 
 value of $s_i$, and $M^{i}_1$ realizes the counter-clockwise rotation by the same angle.
Then for all $\sigma_1, \ldots, \sigma_n \in \{0,1\}^n$ we have that $W = \{s_i \mid \sigma_i = 1\}$ satisfies $\sum W = \sum (S \setminus W)$ iff $M^1_{\sigma_1} \cdots M^n_{\sigma_n}$ equals the identity matrix.

In the reduction we choose initial vector $\iota = (1/2, 1/2)$, final vector $f = (1/2, 1/2)^T$ and threshold $\lambda = 1/2$.
In this way, the chosen chain of rotation matrices $M^1_{\sigma_1} \cdots M^n_{\sigma_n}$ is applied to $f$ before it is checked whether the resulting vector $v$ satisfies $(1/2, 1/2)\cdot v \geq 1/2$. This is  the case iff $v=(1/2,1/2)^T$ and hence iff $M^1_{\sigma_1} \cdots M^n_{\sigma_n}$ is equal to the identity matrix. For an illustration of this idea, see Figure \ref{fig:twoMCP}.

\begin{figure}
  \centering
  \resizebox{0.4\linewidth}{!}{
    \includegraphics{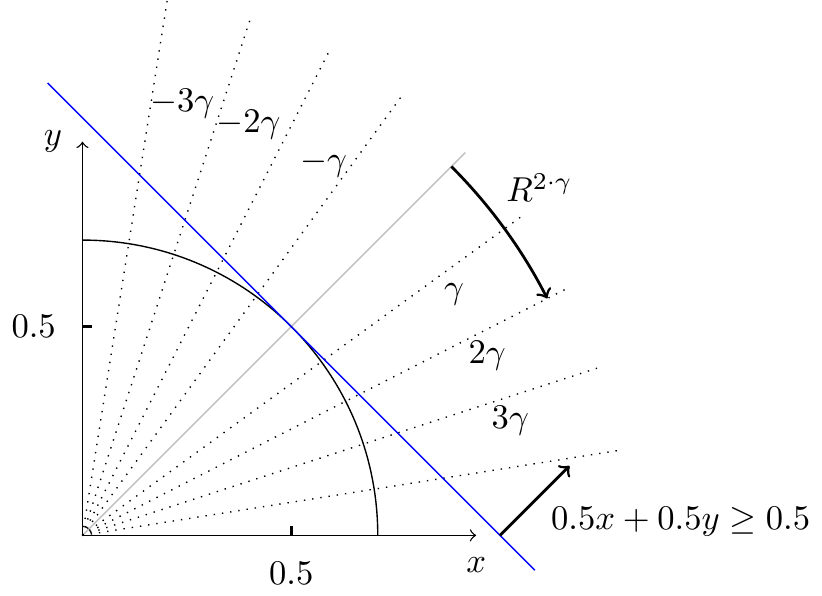}
  }
  \caption{Sketch for~\Cref{prop:mcp}. Matrices are rotations by angles that are multiples of $\gamma$. In order to satisfy the threshold, the final point needs to lie in the halfspace on the right of the blue line. This is only possible if the angles of the rotation matrices sum up to zero.}\label{fig:twoMCP}
\end{figure}

Rotation matrices, however, may have irrational entries in general.
It is shown in \cite{CannyDR1992} that for any rational rotation angle $\varphi$ and $\epsilon \in \mathbb{Q}_{>0}$, a rotation matrix $R^{\theta}$ (rotating by angle $\theta$) with rational entries can be computed in time polynomial in $\log (1/ \epsilon)$ such that $|\varphi - \theta| < \epsilon$.
%\todoSimon{Rev. 2: how is $\varphi$ encoded?}
Using such rational matrices which rotate by approximately the desired angles and by estimating the resulting precision of the matrix multiplication, we can provide a slightly smaller threshold $\lambda^\prime< \lambda$ to complete the reduction from the partition problem to the 2-MCP with rational matrices (a detailed proof can be found in Appendix \ref{app:proofsMCP}).

\begin{restatable}{proposition}{twoMCP}
  \label{prop:mcp}
  The two-dimensional matrix-pair chain problem ($2$-MCP) is NP-complete.
\end{restatable}

The proof of~\Cref{prop:mcp} depends crucially on the fact that negative numbers are allowed in the $2$-MCP, as we would not be able to use the rotation matrices otherwise.
To encode the $2$-MCP into the \emph{nonnegative} $3$-MCP, one can embed the two-dimensional dynamics of a given instance of $2$-MCP into the two-dimensional plane with normal vector $(1,1,1)$ in three dimensions.
To obtain nonnegative matrices, we push vectors further into the direction $(1,1,1)$ at each matrix multiplication step, while preserving the original $2$d-dynamics when projecting onto the subspace orthogonal to $(1,1,1)$. 

To sketch this idea in more detail, let  $(M_0^1,M_1^1),\ldots,(M_0^n,M_1^n)$, with $M_i^j \in \mathbb{Q}^{2 \times 2}$ for all $(i,j) \in \{0,1\} \times \{1,\ldots,n\}$, and $\iota \in \mathbb{Q}^{1 \times 2},f \in \mathbb{Q}^{2 \times 1}$ be an instance of $2$-MCP.
For some $\kappa\in \mathbb{Q}$, we define 
  \begin{equation*}
    N_i^j = 
    B \begin{pmatrix}
      M_i^j & \mathbf{0} \\ \mathbf{0} & \kappa
    \end{pmatrix} B^{-1}, \quad
    \iota' = 
    \begin{pmatrix}
      \iota & \kappa
    \end{pmatrix} B^{-1}, \quad
    f' =
    B\begin{pmatrix}
    f \\ \kappa
    \end{pmatrix} \;\; \text{ and } \;\; \lambda' = \lambda + \kappa^{n+2}
  \end{equation*}
where we use the matrix 
 \[B = 
  \begin{pmatrix}
    1 & 1 & 1 \\
    -1 & 1 & 1 \\
    0 & -2 & 1
  \end{pmatrix}
  \qquad \text{with inverse } \quad B^{-1} = 
  1/6\cdot \begin{pmatrix}
    3 & -3 & 0 \\
    1 & 1 & -2 \\
    2 & 2 & 2
  \end{pmatrix}
  \]
  to change the basis. Note that the columns of $B$ are orthogonal to each other and that the third standard basis vector is mapped to $(1,1,1)$ under the change of basis.
For any $\sigma_1,\ldots,\sigma_n \in \{0,1\}^n$, it is now easy to compute that:
  \begin{align*}
    \iota^\prime \cdot N_{\sigma_1}^1 \cdots N_{\sigma_n}^n \cdot f^\prime =
     \iota\cdot  M_{\sigma_1}^1 \cdots M_{\sigma_n}^n \cdot f + \kappa^{n+2}.
   \end{align*}
So, the constructed instance of the 3-MCP is a yes-instance if and only if the original instance of the 2-MCP is one.
By choosing $\kappa$ large enough, we furthermore can make sure that all matrices $N_j^i$ are nonnegative. This completes the proof of the following proposition. Details can be found in Appendix \ref{app:proofsMCP}.

\begin{restatable}{proposition}{threeMCP}
  \label{prop:nonnegativematrixchain}
  The nonnegative three-dimensional matrix-pair chain problem (nonnegative $3$-MCP) is NP-complete.
\end{restatable}

\subsection{Hardness of the witness problem}\label{sec:hardness_witness}

\begin{figure}
  \centering
  \includegraphics[scale=0.9]{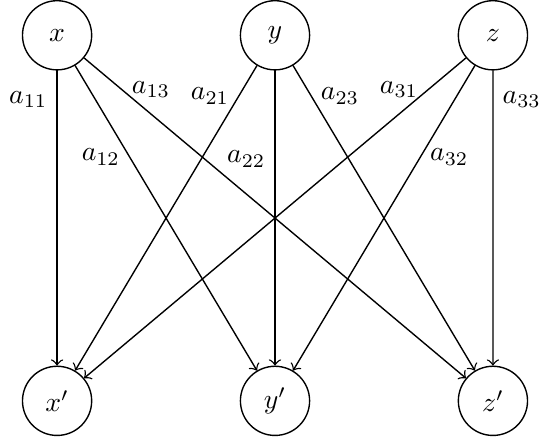}
  \caption{A gadget to encode matrix multiplication. Let $M$ be a substochastic matrix with entries $(M)_{ij} = a_{ij} \in \mathbb{Q}_{\geq 0}$.
  If the probability of states $(x',y',z')$ to reach some goal state is $(v'_x,v'_y,v'_z)$, then these probabilities in states $(x,y,z)$ are
  $M \cdot (v'_x, v'_y, v'_z)^T$.
  We will abbreviate this gadget (i.e. the transitions) by a double-arrow annotated with the matrix ($\xRightarrow{M}$).}
  \label{fig:matrmult}
\end{figure}

The aim of this section is to prove that the witness problem is NP-hard for Markov chains with bounded path-partition width.
The proof goes by a polynomial reduction from the nonnegative 3-MCP.
Let $(M_0^1,M_1^1),\ldots, (M_0^n,M_1^n)$, $\iota,f$ and $\lambda$ be an instance of this problem. For technical reasons explained later, we assume that all entries of the input matrices and vectors are in the range $[1/12 - \epsilon, 1/12]$ for some $\epsilon$ that satisfies:
\begin{equation}
  \label{eqn:epsbound}
  0 < 12 \epsilon < 1/3 \cdot \big(1/12 - \epsilon\big)^{n+2}
\end{equation}
In~\Cref{sec:matrassumptions} we show that the nonnegative $3$-MCP problem remains NP-hard under these assumptions.

\subsubsection{Structure of the reduction}

As a first step,~\Cref{fig:matrmult} shows how one can encode the multiplication of a (suitable) vector with a (suitable) matrix in a Markov chain.
Using this gadget,~\Cref{fig:mainstruct} shows the main structure of the reduction from the nonnegative 3-MCP.
The double arrows represent instances of the gadget from~\Cref{fig:matrmult}.
The initial distribution assigns probability $\iota(x)$ to both states $\rst{1}{x}$ and $\lst{1}{x}$, and similarly for $y$ and $z$.
The final edge from state $x_{n{+}1}$ to $\goal$ has probability $f(x)$, and analogously for all other states on the final layer.
The above assumption on the entries of the matrices and vectors guarantees that the sums of initial probabilities and outgoing probabilities of any state are below $1$.
Hence the result is indeed a Markov chain which we call $\M_1$.
Furthermore, the directed tree-partition width and directed path-partition width of $\M_1$ are both independent of the 3-MCP instance (for the proof, see \Cref{app:hardness_witness}):
\begin{restatable}{lemma}{constrWidth}
  \label{lem:constrwidth}
  $\dppw(\M_1) = \dtpw(\M_1) = 6$.
\end{restatable}

\begin{figure}
  \centering
  \resizebox{0.8\linewidth}{!}{
    \includegraphics{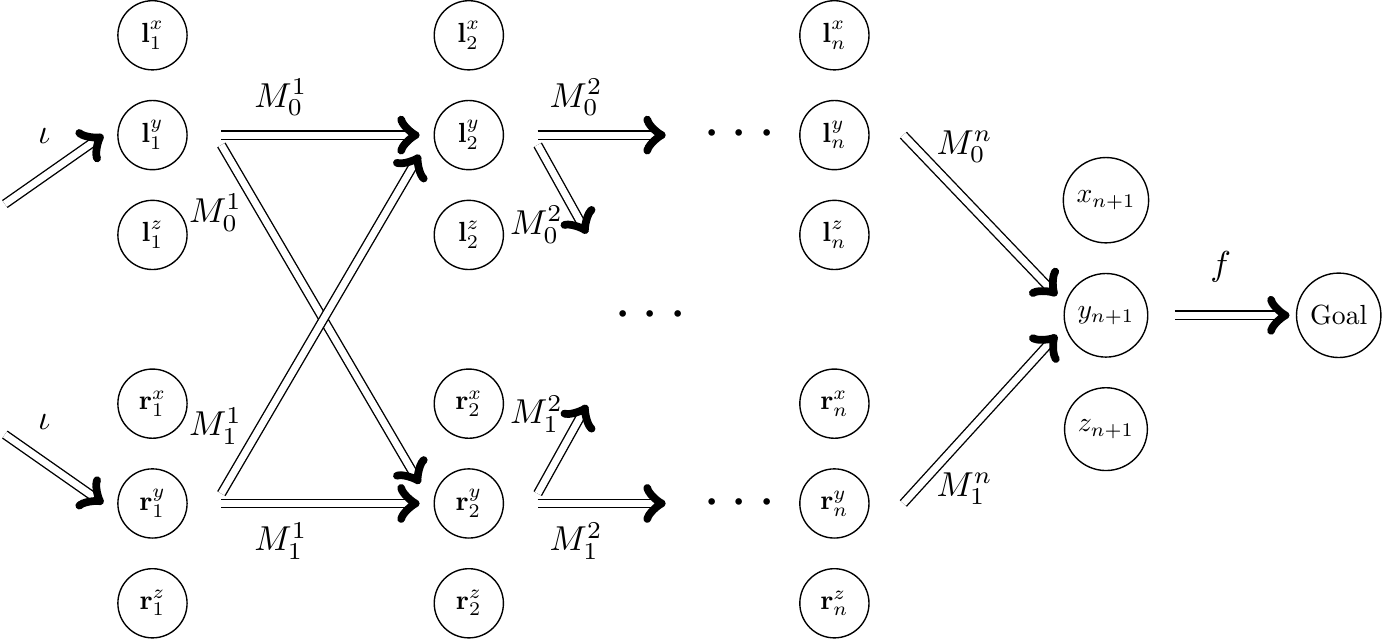}
  }
  \caption{Main structure of the reduction. The \emph{good subsystems} contain either the states $\{\lst{j}{x},\lst{j}{y},\lst{j}{z}\}$ or $\{\rst{j}{x},\rst{j}{y},\rst{j}{z}\}$ for each layer $j$ and thereby correspond to a choice $\sigma_1 \ldots \sigma_n$ in the matrix-pair chain problem.}
  \label{fig:mainstruct}
\end{figure}
Let $\llay{i} = \{\lst{i}{x},\lst{i}{y},\lst{i}{z}\}$ and $\rlay{i} = \{\rst{i}{x},\rst{i}{y},\rst{i}{z}\}$.
A subsystem $S' \subseteq S$ is called \emph{good} if it includes the states $\{x_{n+1},y_{n+1},z_{n+1}\}$ and for all $1 \leq i \leq n$:
\[\text{\emph{either}} \quad \llay{i} \subseteq S' \;\; \text{and} \;\; \rlay{i} \cap \ S' = \varnothing \qquad \text{\emph{or}} \qquad \rlay{i} \subseteq S' \;\; \text{and} \;\; \llay{i} \cap \, S' = \varnothing\]
That is, $S'$ chooses  exactly one of the sets $\llay{i}$ and $\rlay{i}$ in each layer $i$, with $1 \leq i \leq n$.
Good subsystems have exactly $3n + 4$ states (including $\goal$).
Subsystems that are not good are called \emph{bad}.
There is a one-to-one correspondance between good subsystems and matrix sequences in the matrix-pair chain problem.
For a given sequence $\boldsymbol{\sigma} = \sigma_1,\ldots,\sigma_n \in \{0,1\}^n$, let:
\[S_{\boldsymbol{\sigma}} = \{x_{n{+}1},y_{n{+}1},z_{n{+}1}\} \cup \bigcup_{\substack{1 \leq i \leq n\\\sigma_i = 0}} \llay{i} \cup \bigcup_{\substack{1 \leq i \leq n\\\sigma_i = 1}} \rlay{i}\]
In the following we denote by $\Pr^{\M_1}_{S'}(\lozenge \goal)$ the probability of reaching $\goal$ under the subsystem induced by $S'$ in $\M_1$.
The following lemma shows that the probability of reaching $goal$ in a good subsystems coincides with the corresponding matrix product (see~\Cref{app:hardness_witness} for the proof).
\begin{restatable}{lemma}{goodSubsysyProp}
  \label{lem:goodsybsysprop}
  For all $\boldsymbol{\sigma} \in \{0,1\}^n$ we have: $\Pr^{\M_1}_{S_{\boldsymbol{\sigma}}}(\lozenge \goal) = \iota \cdot M_{\sigma_1}^1 \cdots M_{\sigma_n}^n \cdot f$.
\end{restatable}
It follows that the 3-MCP reduces to deciding whether there exists a \emph{good} subsystem whose probability to reach goal is at least $\lambda$.
However, it could still be the case that while the 3-MCP instance is a no-instance, there is some \emph{bad} subsystem of size $n$ that satisfies the threshold condition.
We now show how $\M_1$ can be adapted such that the subsystems of size $3n + 4$ with greatest probability are the good ones.

\subsubsection{Interconnecting states}
\label{sec:interconnectingstates}
The idea is to make sure that bad subsystems have decisively less probability to reach $goal$.
To this end we adapt the matrix multiplication gadget from~\Cref{fig:matrmult} such that removing any state leads to a large drop in probability.
This is achieved by adding a cycle which connects the upper states, as shown in~\Cref{subfig:newmultgen}.
The states $x_i,y_i,z_i$ represent one of the triples $\lst{i}{x},\lst{i}{y},\lst{i}{z}$ or $\rst{i}{x},\rst{i}{y},\rst{i}{z}$, and likewise for $x_{i+1},y_{i+1},z_{i+1}$.
The probability of staying inside the cycle is $\gamma$ in each state and the matrix that contains the pairwise probabilities of reaching the states $(x_{i+1},y_{i+1},z_{i+1})$ from states $(x_i,y_i,z_i)$ is:
\begin{equation}
  \label{eqn:gammacyclesprob}
  M = 
  \underbrace{
    \frac{1-\gamma}{1-\gamma^3} \cdot
    \begin{pmatrix}
      1 & \gamma & \gamma^2 \\
      \gamma^2 & 1 & \gamma \\
      \gamma & \gamma^2 & 1 \\
  \end{pmatrix}}_{R} \
  M' \quad \text{ with } \;
  R^{-1} = 
  \frac{1}{1{-}\gamma}
  \begin{pmatrix}
    1 & -\gamma & 0 \\
    0 & 1 & -\gamma \\
    -\gamma & 0 & 1 \\
  \end{pmatrix}
\end{equation}

The edges between the last layer $(x_{n+1},y_{n+1},z_{n+1})$ and $\goal$ are adapted in a similar way.
\begin{figure}
  \begin{subfigure}[t]{.5\linewidth}
    \hfill
  \resizebox{0.9\linewidth}{!}{
    \includegraphics{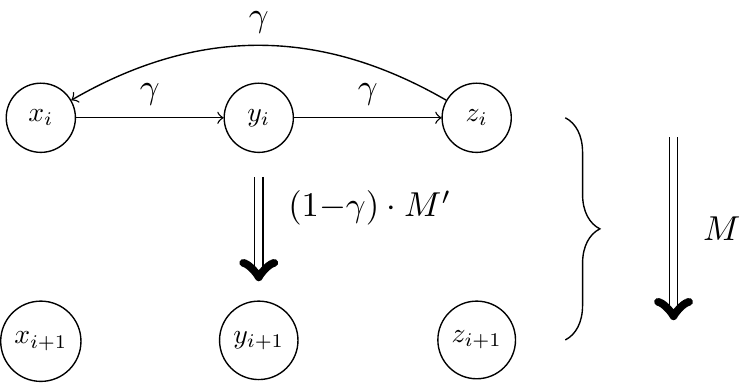}
  }
  \caption{}
  \label{subfig:newmultgen}
  \end{subfigure} \hfill
  \begin{subfigure}[t]{.5\linewidth}
    \hfill
  \resizebox{0.9\linewidth}{!}{
    \includegraphics{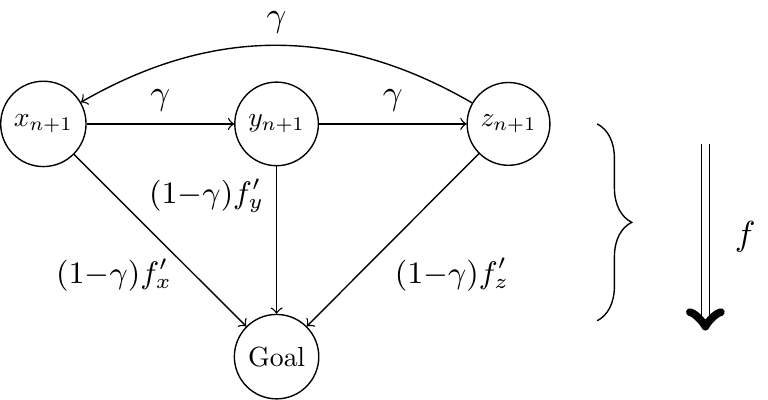}
  }
  \caption{}
  \label{subfig:newmultlastlay}
  \end{subfigure}
  \caption{A $\gamma$-cycle is added to the upper states of the matrix multiplication gadget (see~\Cref{fig:matrmult}) to make sure that removing any state on the cycle leads to a significant drop in probability.
    \Cref{subfig:newmultgen} shows the construction used in all but the last layer, which is handled by the construction in~\Cref{subfig:newmultlastlay}.
    In~\Cref{subfig:newmultgen}, the matrix $M'$ is chosen such that the probability of reaching $(x_{i+1},y_{i+1},z_{i+1})$ is $\theta \cdot M$, where $\theta$ is any initial distribution on states $(x_{i},y_{i},z_{i})$, and similarly in~\Cref{subfig:newmultlastlay}.}
  \label{fig:matrmult2}
\end{figure}
Let us assume that we are given a matrix $(M)_{ij} = a_{ij}$ (this will be one of the input matrices of the 3-MCP) and we want to find $M'$ such that the gadget from \Cref{subfig:newmultgen} realizes the matrix multiplication $M$.
In other words, we want the probability to reach $x_{i+1}$ from $x_i$ to be exactly $a_{11}$, and similarly for the other states.
Solving the equation above for $M'$ gives:
\begin{equation}
\label{eqn:gammamatrix}
M' = R^{-1} \cdot M =
\frac{1}{1{-}\gamma}
{\renewcommand*{\arraycolsep}{5pt}\begin{pmatrix}
  a_{11} - \gamma a_{21} & a_{12} - \gamma a_{22} & a_{13} - \gamma a_{23} \\
  a_{21} - \gamma a_{31} & a_{22} - \gamma a_{32} & a_{23} - \gamma a_{33} \\
  a_{31} - \gamma a_{11} & a_{32} - \gamma a_{12} & a_{33} - \gamma a_{13} \\
\end{pmatrix}}
\end{equation}
We choose $\gamma$ to satisfy:
\begin{equation}
  \label{eqn:gammachoice}
  12 \epsilon < 1{-}\gamma < 1/3 \cdot \big(3(1/12 - \epsilon)\big)^{n+2}
\end{equation}
which is possible due to the assumption of~\Cref{eqn:epsbound}.
This makes sure that all entries of $M'$ are nonnegative.
The argument uses that the entries $a_{ij}$ are assumed to be in the range $[1/12 - \epsilon, 1/12]$:
\[ \frac{1}{1-\gamma}(a - \gamma a') \geq \frac{1}{1-\gamma}(1/12 - \epsilon - \gamma/12) = 1/12 - \frac{\epsilon}{1-\gamma} > 0 \;\; \text{ for all entries } a,a' \text{ of } M\]
where the last inequality follows from $12 \epsilon < 1 {-} \gamma$.
Furthermore, we have:
\[ \frac{1}{1-\gamma}(a - \gamma a') \leq \frac{1}{1-\gamma}(1/12 - \gamma (1/12 - \epsilon)) = 1/12 + \frac{\gamma \epsilon}{1-\gamma} < 1/6 \;\; \text{ for all entries } a,a' \text{ of } M\]
where the last inequality follows from $\gamma < 1$ and $12 \epsilon < 1 - \gamma$, which is equivalent to $\epsilon / (1 - \gamma) < 1/12$.
The fact that $1/6$ is an upper bound on all entries of $M'$ implies that using the gadgets from~\Cref{fig:matrmult2} in the reduction yields a DTMC (observe that all states in~\Cref{fig:mainstruct} have at most 6 outgoing edges).

We call the Markov chain that is obtained by adding the $\gamma$-cycles and adapting the probabilities as discussed above $\M_2$.
The construction ensures that the good subsystems (defined as for $\M_1$) have the same probability to reach $\goal$ in both DTMCs, and hence~\Cref{lem:goodsubsysoptimal} holds as well for $\M_2$.
The main point of adding the $\gamma$-cycles was to make sure that if one state from $\{x',y',z'\}$ is excluded in a subsystem, then the probability of any state in $x_i,y_i,z_i$ to reach the next layer drops significantly.
Now this value is indeed bounded by $(1 + \gamma + \gamma^2) \cdot (1-\gamma) < 3 \cdot (1-\gamma)$ (as $\gamma < 1$).
In a bad subsystem, both $\gamma$-cycles are interrupted on some layer.
Hence, the probability to reach $\goal$ is less than $3 \cdot (1-\gamma)$. This value in turn is less than $\big(3(1/12 - \epsilon)\big)^{n+2}$ by \Cref{eqn:gammachoice}.
On the other hand,
the fact that all entries of matrices $M_i^j$ (with $1 \leq i \leq n$ and $0 \leq j \leq 1$) and vectors $\iota,f$ have value at least $1/12 - \epsilon$ implies that $\big(3(1/12 - \epsilon)\big)^{n+2}$ is a lower bound on the reachability probability that is achieved by any good subsystem. A detailed discussion of these facts that lead to the following lemma can be found in \Cref{app:goodsubsysoptimal}.

\begin{restatable}{lemma}{goodsubsystemoptimal}
  \label{lem:goodsubsysoptimal}
  Let  $S_1$ and $S_2$ be a  subsystems of $\M_2$ with $3n + 4$ states. If $S_1$ is bad and $S_2$ good, then
  \[\Pr^{\M_2}_{S_1}(\lozenge \goal) \leq \Pr^{\M_2}_{S_2}(\lozenge \goal)\]
\end{restatable}

Finally, note that the directed path-partition-width and tree-partition-width of $\M_2$ is the same as of $\M_1$, as $\M_2$ includes more edges but still allows the directed path-partition which partitions states along the layers.
Hence we have $\dtpw(\M_2) = \dppw(\M_1) = 6$. %with the same argument as in~\Cref{lem:constrwidth}.
Together with~\Cref{lem:goodsubsysoptimal}, \Cref{lem:goodsybsysprop} and the fact that the probabilities of good subsystems in $\M_1$ and $\M_2$ coincide, this proves:
\begin{theorem}
  \label{thm:witnesshard}
  The witness problem is NP-hard for Markov chains with $\dppw{} = 6$ (and hence also for Markov chains with $\dtpw{} = 6$).
\end{theorem}

\section{A dedicated algorithm for MDPs and a given tree partition}
\label{sec:algorithm}
This section introduces an algorithm (\Cref{alg:treelike}) that computes a minimal witnessing subsystem using a given \emph{directed tree partition} of the system.
The main idea is to proceed bottom-up along the induced tree order and enumerate partial subsystems for each block and to compute the values achieved by the ``interface'' states for each partial subsystem.
Interface states are those that have incoming edges from the predecessor block in the tree partition.
A domination relation between partial subsystems is used to prune away all partial subsystems that do not need to be considered further up, as a ``better'' one exists.

Let $\M = (S,\Act,P,\iota)$ be a fixed MDP for the rest of this section, and $\P = \{B_1, \ldots, B_n\}$ be a directed tree partition of $\M$.
We will assume that for all $B \in \P$ we have $B \subseteq \goal$ or $B \cap \goal = \varnothing$.
This is not a real restriction as states in $\goal$ are trap states, which means that they can always be moved to a separate block.
Furthermore, we assume that all initial states are in the root block of the tree partition.
We denote by $\sucs(B_i) \subseteq \P$ the children of $B_i$ in the associated tree order, and by $\pre(B_i) \in \P$ the unique parent of $B_i$.
For each block $B_i$ we denote by $\interface(B_i)$ the states in $B_i$ which have some incoming edge from a state in $\pre(B_i)$ or are initial.
Using this notion we define $\outint(B_i) = \bigcup_{B \in \sucs(B_i)} \interface(B)$.
We denote by $\closure(B_i)$ the union of blocks $B \in \P$ such that $B$ is reachable from $B_i$ in the tree order.

For a given partial function $f$ from $S$ to $[0,1]$ and subset $S' \subseteq S$ we consider the MDP $\M_{S'}^f$ constructed as follows.
In the subsystem $\M_{S'}$ induced by $S'$ remove all outgoing edges from states $s \in \dom(f)$ (the domain of $f$) and replace them by an action with a single transition to (some state in) $\goal$ with probability $f(s)$ (resulting again in a sub-stochastic MDP).
We define for each state $q \in S'$:
\[\minval_{S'}^f(q) = \prb^{\min}_{\M_{S'}^f(q)}(\lozenge \goal) \qquad \text{and} \qquad \maxval_{S'}^f(q) = \prb^{\max}_{\M_{S'}^f(q)}(\lozenge \goal)\]
We write $\minval_{S'}$ or $\maxval_{S'}$ for the respective values in the unchanged MDP $\M_{S'}$.
The following lemma shows that to compute the values of states in $B_i$ under any subsystem, one can first compute the values of states in $\outint(B_i)$, then replace the edges of those states by an edge to $\goal$ carrying this value, and finally compute the values for states in $B_i$ in the adapted system.
\begin{restatable}{lemma}{CompositionLemma}
  \label{lem:composition}
  Let  $S' \subseteq S$, $S_2 = S' \cap (\closure(B_i) \setminus B_i)$ (with $1 \leq i \leq n$) and $S_1 = S' \setminus S_2$.
  Let $v \in \{\maxval,\minval\}$ and define $f$ over domain $\outint(B_i)$ by: $f(q) = v_{S_2}(q)$ for all $q \in \outint(B_i)$.
  Finally, let $S_1' = S_1 \cup \outint(B_i)$.

  Then, for all $q \in S_1'$:
  \[v_{S'}(q) =  v_{S_1'}^{f}(q)\]
\end{restatable}

\begin{figure}[tbp]
  \centering
  \includegraphics{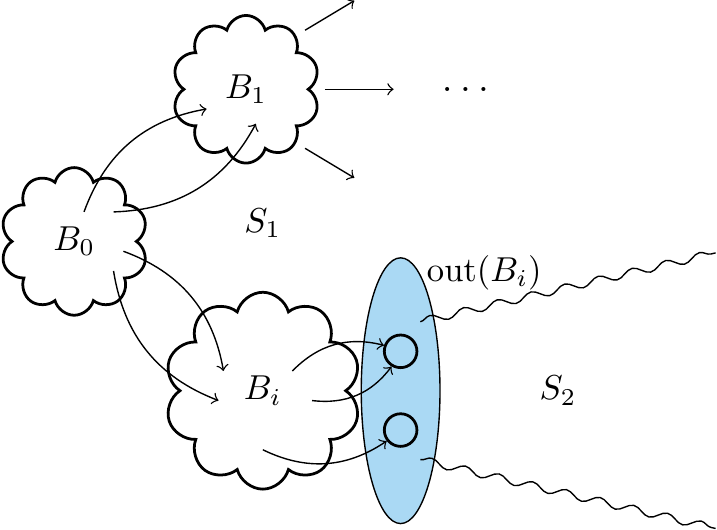}
  \caption{Visualizuation of the situation in~\Cref{lem:composition}. Some block $B_i$ is fixed, and $\outint(B_i)$ are the states outside of $B_i$ which are reachable from some state in $B_i$. The state set is partitioned into the sets $S_1$, which includes $B_i$ and all states that are not reachable from $B_i$, and $S_2$, which includes all states reachable from $B_i$ but excluding $B_i$. Additionally, in~\Cref{lem:composition} all of these sets are intersected with a set $S'$ in~\Cref{lem:composition}, which represents some subset of the entire system (this is not depicted here).}
\end{figure}

\subsection{The domination relation}
\label{sec:dominationrel}
In the following, the vector $f$ can be thought of as an assumption on the value that is achieved in states in $\dom(f)$.
Different partial subsystems of the system in a subtree will correspond to different vectors $f$, where $\dom(f)$ are the interface states.
For two partial functions $f_1,f_2 : S \to [0,1]$ such that $\dom(f_1) = \dom(f_2)$ we define $(f_1 + f_2)(q) = f_1(q) + f_2(q)$, for $q \in \dom(f_1)$, and write $f_1 \leq f_2$ to mean $f_1(q) \leq f_2(q)$ for all $q \in \dom (f_1)$. By $\mathbf{1}$ we denote the constant $1$-function with suitable domain.
\begin{restatable}{lemma}{propertiesValuefuncs}
  \label{lem:propertiesValuefuncs}
  Let $T \subseteq S$, $v \in \{\maxval,\minval\}$ and $f_1,f_2 : S \to [0,1]$ be partial functions such that $\dom(f_1) = \dom(f_2)$ and let $I \subseteq T$ be a set of states that cannot reach $\goal$ without seeing $\dom(f_1)$ in $\M$.
  Then, for all $q \in I$:
  \begin{enumerate}
  \item $f_1 \geq f_2  \implies v_{T}^{f_1}(q) \geq v_{T}^{f_2}(q)$,
  \item for all $a \in \mathbb{Q}_{\geq 0}$ such that $a \cdot f \leq \mathbf{1}$: $ \quad a \cdot v_{T}^{f}(q) = v_{T}^{a \cdot f}(q)$,
  \item if $f_1 + f_2 \leq \mathbf{1}$, then: $\maxval_{T}^{f_1 {+} f_2}(q) \leq \maxval_{T}^{f_1}(q) + \maxval_{T}^{f_2}(q)$.
    \end{enumerate}
\end{restatable}

Let $v \in \{\maxval,\minval\}$ be fixed for the remainder of this section.
For a given set $I \subseteq S$, we denote the states \emph{reachable} from $I$ in the underlying graph of $\M$ by $\reach(I)$.
A \emph{partial subsystem} for $I$ is a set $T \subseteq \reach(I)$ and the \emph{$I$-point} corresponding to $T$ is defined to be the vector $\pointval_I(T) \in \mathbb{Q}^{I}$ with $\pointval_I(T)(q) = v_{T}(q)$ if $q \in T \cap I$ (where $v_{T}$ is the value vector under subsystem $T$) and $\pointval_I(T)(q) = 0$ if $q \in I \setminus T$.
%%  $(v_{T})|_{I}$, that is, the restriction of the value vector under subsystem $T$ onto the states $I$.
%% The $I$-point corresponding to the partial subsystem $T$ will be denoted by .
Let $\pi$ be the function which collects all possible projections of a vector $\theta \in \mathbb{Q}^{I}$ onto a subset of the axes as follows:
\[\pi(\theta) = \{ \pi(\theta,D) \mid D \subseteq I\} \qquad\qquad \text{and} \qquad \qquad
\pi(\theta,D)(x) = 
\begin{cases}
  \theta(x) & x \in D \\
  0 & \text{otherwise}
\end{cases}
\]
If a partial subsystem $T$ for $I$ includes no $\goal$ state, then $\pointval_I(T)$ will be zero in all entries, and hence we will not be interested in such $T$.
However, $T$ does not have to include all $\goal$ states reachable from $I$.

The core of~\Cref{alg:treelike} is the domination relation (see~\Cref{fig:dominationexample}) which is used to discard partial subsystems. Without it the algorithm would amount to an explicit enumeration of all subsystems.

\begin{definition}
  \label{def:domination}
  Let $I \subseteq S$ and $\{T\} \cup \mathcal{S}$ be a set of partial subsystems for $I$.
  We say that $\mathcal{S}$ dominates $T$ if there exists $\mathcal{S}' \subseteq \mathcal{S}$ such that
  \begin{enumerate}
  \item For all $T' \in \mathcal{S}'$ we have $|T'| \leq |T|$, and
  \item $\pointval_I(T)$ is a convex combination of $\ \bigcup\{\pi(\pointval_I(T')) \mid T' \in \mathcal{S}'\}$.
  \end{enumerate}
  We say that $\mathcal{S}$ \emph{strongly} dominates $T$ if there exists a \emph{singleton} set $\mathcal{S}' \subseteq \mathcal{S}$ such that $\mathcal{S}'$ dominates $T$.
\end{definition}
\begin{figure}[tbp]
  \centering
  \includegraphics[width=0.27\textwidth]{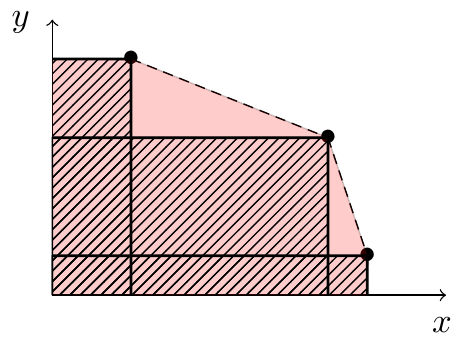}
  \caption{The black points represent three partial subsystems for $I =\{x,y\}$ via their $I$-points. The red area indicates $I$-points of partial subsystems which are dominated by these three points, while the dashed area indicates the partial subsystems \emph{strongly} dominated by one of them. The size of the partial subsystems is not considered here, but is important in general (see~\Cref{def:domination}).}
  \label{fig:dominationexample}
\end{figure}
%Partial subsystems  for $I = \{x,y\}$ represented by their value vectors in $[0,1]^{\{x,y\}}$. 
\newpage
\begin{restatable}{lemma}{dominationRel}
\label{lem:dominationrel}
Let $S' \subseteq S$, $S_2 = S' \cap (\closure(B_i) \setminus B_i)$ (for some $1 \leq i \leq n$) and $S_1 = S' \setminus S_2$.
Furthermore, let $I = \outint(B_i)$ and $\mathcal{S}$ be a set of partial subsystems for $I$.
\begin{enumerate}
\item If $\M_{S'}$ is a witnessing subsystem for $\prb^{\min}(\lozenge \goal) \geq \lambda$ and $\cal S$ \emph{strongly dominates} $S_2$, then there is a $T \in \mathcal{S}$ such that $\M_{S_1 \cup T}$ is a witnessing subsystem for $\prb^{\min}(\lozenge \goal) \geq \lambda$ and $|T| \leq |S_2|$.
\item If $\M_{S'}$ is a witnessing subsystem for $\prb^{\max}(\lozenge \goal) \geq \lambda$ and $\cal S$ \emph{dominates} $S_2$, then there is a $T \in \mathcal{S}$ such that $\M_{S_1 \cup T}$ is a witnessing subsystem for $\prb^{\max}(\lozenge \goal) \geq \lambda$ and $|T| \leq |S_2|$.
\end{enumerate}
\end{restatable}

\Cref{alg:convhull} details how the domination relation can be computed using an incremental convex-hull algorithm.
The ConvexHull object that is used in line 3 allows to add points incrementally, and stores the vertices of the convex hull of points added so far in the field \emph{vertices}.
The convex hull of $a$ points in $d$ dimensions can be computed in $O(a \cdot \log a + a^{\lfloor d/2 \rfloor})$~\cite{Chazelle1993}.
In our case $d$ corresponds to the number of interface states $|I|$, and as a number of dedicated and fast algorithms exist to compute the convex hull in low dimensions\cite{Chan1996,BarberDH1996,Graham1972} tree partitions with few interface states in each block are desirable.

\SetKw{KwForIn}{in}
\SetKw{KwNot}{not}
\SetKw{KwSuchThat}{such that}
\SetKw{KwFor}{for}

\begin{algorithm}[tbp]
  \footnotesize
  \SetAlgoLined
  \KwIn{Set of partial subsystems $\mathcal{S}$ for $I$, with $I \subseteq S$.}
  \tcc{Group partial subsystems by their size.}
  $\mathcal{S}[k] := \{S' \in \mathcal{S} \mid |S'| = k\}$ \\
  $m := \max \{|S'| \mid S' \in \mathcal{S}\}$ \\
  \tcc{Initialise an empty ConvexHull object}
  $\mathcal{H} := $ ConvexHull() \\
  \For{$k = 1$ \KwTo $m$ \label{alg:convhull:forloop}}{
    \tcc{Compute projections of value vectors in $\mathcal{S}[k]$.}
    $\Pi := \bigcup\{\pi(\pointval_I(S')) \mid S' \in \mathcal{S}[k]\}$\\
    \tcc{Add $\Pi$ to the incremental ConvexHull object.}
    $\mathcal{H}$.addPoints($\Pi$)\\
    \tcc{Remember only subsystems in $\mathcal{S}[k]$ that are vertices of $\mathcal{H}$.}
    $R := R \cup \{S' \in \mathcal{S}[k] \mid \pointval_I(S') \in \mathcal{H}$.vertices $\}$ \label{alg:convhull:updateR}\\
  }
  \KwRet{R}
  \caption{removeDominated}
  \label{alg:convhull}
\end{algorithm}

\begin{restatable}{lemma}{removeDominated}
  \label{lem:removedom}
  Let $\mathcal{S}$ be a set of partial subsystems for $I \subseteq S$ and $R =$ removeDominated$({\cal S})$.
  Then,
  \begin{itemize}
  \item for any $T \in \mathcal{S} \setminus R$ it holds that $R$ dominates $T$.
  \item no $T \in R$ is dominated by $R \setminus \{T\}$.
  \end{itemize}
\end{restatable}

In order to avoid enumerating all subsets of a block we first apply a filter based on a Boolean condition.
It requires that any state in the subset either is an interface state or has a predecessor in the subset.
Likewise it should either have a successor in the subset, or an outgoing edge to another block.
Consider the following Boolean formula with variables in $S$:
\begin{align*}
  \phi(B_i) = \bigwedge_{s \not\in \interface(B_i)} \left( s \rightarrow \bigvee_{s' \in \Pre(s)} s' \right) \land \bigwedge_{s \not\in \out(t_i)} \left( s \rightarrow \bigvee_{s' \in \Post(s)} s' \right)
\end{align*}
where $\out(B_i) = \{s \in B_i \mid \Post(s) \setminus B_i \neq \varnothing\}$.
Now any partial subsystem $S'$ such that $S' \cap B_i$ is not a model of $\phi(B_i)$ is dominated by another subsystem, which one gets by removing unnecessary states.

\subsection{An algorithm based on the domination relation}

\begin{algorithm}[tbp]
  \footnotesize
  \SetAlgoLined
  \KwIn{MDP $\M$, directed tree partition $\P$, rational $\lambda$}
  \KwOut{Minimal witnessing subsystem for $\prb^{\max}_{\M}(\lozenge \goal) \geq \lambda$.}
  \tcc{Bottom-up traversal of the tree partition.}
  \For{B \KwForIn reverse(topologicalSort($\P$)) \label{alg:treelike:revtop}}{
    $I := \interface(B)$\\
    $O := \outint(B)$ \\
    \tcc{Consider only subsets of $B$ that satisfy $\phi(B)$}
    \For{$S_B \subseteq B$ \KwSuchThat $S_B \models \phi(B)$ \label{alg:treelike:BDD}}{
      \tcc{Consider each combination of partial subsystems of the children of $B$.}
      \For{$(S',\pointval_O(S'))$ \KwForIn successorPoints(\emph{\partsubsysmap},B) \label{alg:treelike:succpoints}}{
        $f := \pointval_O(S')$\\
        \tcc{The new partial subsystem $S_{new}$ for $I$ combines $S_B$ and $S'$.}
        $S_{new} := S_B \cup S'$\\
        $\pointval_I(S_{new}) := (\maxval_{S'}^{f})|_I$ \label{alg:treelike:computeval}\\
        \tcc{Remember the corresponding partial subsystem.}
        $\partsubsysmap[B]$.insert($S_{new}$) \label{alg:treelike:insertp}
      }
      \tcc{Remove dominated points}
      $\partsubsysmap[B]$ := removeDominated($\partsubsysmap[B]$) \label{alg:treelike:removedom}
    }
  }
  \tcc{Here $B_r$ is assumed to be the root of the tree associated with $\P$.}
  \KwRet{argmin$\{|S'|$ \KwFor $S'$ \KwForIn \emph{\partsubsysmap}$[B_r]$ \KwSuchThat $\iota \cdot \pointval_{\supp(\iota)} \geq \lambda \}$}
  \caption{A dedicated algorithm for MDPs using a given directed tree partition.}
  \label{alg:treelike}
\end{algorithm}

\Cref{alg:treelike} computes a minimal witnessing subsystem of $\M$ for $\prb^{\max}_{\M}(\lozenge \goal) \geq \lambda$, using the structure of the tree decomposition $\cal P$.
Witnesses for $\prb^{\min}_{\M}(\lozenge \goal) \geq \lambda$ can be handled by replacing the call to removeDominated in~\Cref{alg:treelike:removedom} by a method which computes the \emph{strong} domination relation (see~\Cref{table:instances} for the possible instances of the algorithm).
Computing the strong domination relation requires checking whether the $I$-point of a new partial subsystem is pointwise smaller than that of any of the given partial subsystems.
The algorithm keeps a map $\partsubsysmap$ from blocks $B \in \mathcal{P}$ to partial subsystems for $\interface(B)$.
This map is populated in a bottom-up traversal of $\mathcal{P}$ (\Cref{alg:treelike:revtop}).
For a given block $B$, the models of $\phi(B)$ (which are subsets of $B$) are enumerated (\Cref{alg:treelike:BDD}).
The method \emph{successorPoints} in~\Cref{alg:treelike:succpoints} returns all partial subsystems for $O = \outint(B)$ which can be obtained by combining partial subsystems in $\partsubsysmap[B_i]$ for all $B_i \in \sucs(B)$. More precisely, if $\sucs(B_i) = \{P_1, \ldots P_k\}$, then:
\[successorPoints(\partsubsysmap,B) = \{ \bigl(T, \pointval_O(T)\bigr) \mid S_1 \in \partsubsysmap[P_1], \ldots, S_k \in \partsubsysmap[P_k]\}\]
where $T = \bigcup_{1 \leq i \leq k} S_i$ and $\pointval_O(T)$ is the vector one gets by concatenating vectors $\pointval_{\interface(P_i)}(S_i)$ (recall that $O = \outint(B) = \bigcup_{1 \leq i \leq k} \interface(P_i)$, and the blocks are pairwise disjoint).
The vectors $\pointval_{\interface(P_i)}(S_i)$ have been computed during a previous iteration of the for loop in~\Cref{alg:treelike:computeval} and are assumed to be in global memory (they are also needed in the algorithm removeDominated).

\begin{restatable}{proposition}{algCorrectness}
  If \Cref{alg:treelike} returns $S'$ on input $(\M,\mathcal{P},\lambda)$, then $S'$ is a minimal witness for $\prb^{\max}_{\M}(\lozenge \goal) \geq \lambda$.
  It returns within exponential time in the size of the input.
\end{restatable}

\paragraph{Additional heuristics to exclude partial subsystems.}%
In addition to the domination relation we propose two conditions on when a partial subsystem can be excluded.
First, suppose we are considering block $B_i$ with interface $I = \interface(B_i)$, and let $\closure(B_i)$ be the union of blocks reachable from $B_i$ and $R = S \setminus \closure(B_i)$.
If using \emph{all states} from $R$ together with a partial subsystem $T$ for $\interface(B_i)$ does not lead to a value above $\lambda$ then $T$ can be excluded.
A sufficient condition for this which can easily be checked is if $\goal \subseteq \closure(B_i)$ holds and the sum of entries of the value vector $\pointval_I(T)$ is less than $\lambda$.

For the second condition, assume that $N$ is an upper bound on the size of a minimal witnessing subsystem (this could have been computed by a heuristic approach) and let $M$ be the length of a shortest path from the initial state into any state of $I$ (these can be computed in advance and in polynomial time).
Now if $M + |T| > N$, then $T$ cannot be part of any minimal witness, and can be excluded.
\subsection{Experimental evaluation}
\label{sec:experiments}
We have implemented~\Cref{alg:treelike} in the tool \textsc{Switss}~\cite{JantschHFB2020} using the convex hull library qhull\footnote{http://www.qhull.org/}.
The experiments were performed on a computer with two Intel E5-2680 8 cores at \(2.70\)\,GHz running Linux, where each instance got assigned a single core, a maximum of 10GiB of memory and 900 seconds.
All datasets and instructions to reproduce can be found in~\cite{Jantsch2021}.
At the moment, the implementation only supports DTMCs as input and only returns the size of a minimal witnessing subsystem.
To evaluate it, we consider the \emph{bounded retransmission protocol} (brp) for file transfers, which is a standard benchmark included in the PRISM benchmark suite\cite{KwiatkowsaNP2012}.
It is parametrized by $N$ (the number of ``chunks'') and $K$ (the number of retransmissions).
We fix $K=1$ but consider increasing values for $N$, yielding instances of size between $151$ ($N = 8$) and $1591$ ($N = 80$) in terms of state numbers.
We consider the probabilistic reachability constraint $\Pr(\lozenge \goal) \geq \lambda$, for varying thresholds $\lambda$ and fixed $\goal$.

The protocol maintains a counter which is only increased up to maximal value $N$, and using this fact one can compute a natural directed path partition for the model which essentially partitions the state space along the possible values of the counter.
The directed path partitions that we get have length $N{+}1$ and constant width $37$.
After filtering out the subsets of a block $B$ that do not satisfy $\phi(B)$ (see~\Cref{sec:dominationrel}) at most $15$ subsets remain.
\Cref{fig:experiments} compares the computation time of~\Cref{alg:treelike} against known \emph{mixed-integer linear programming} (MILP) based approaches to compute minimal witnessing subsystems~\cite{WimmerJABK2012,FunkeJB2020}.
The computation times do not include the generation of the path partition, which is straight forward in this particular case.
In the figure, ``min'' and ``max'' refer to the two MILPs derived from the polytopes $\mathcal{P}^{\min}$ and $\mathcal{P}^{\max}$ defined in~\cite[Lemmas 5.1 and 6.1]{FunkeJB2020}.
To solve the MILPs, we use the solvers \gurobi{}~\cite{gurobi21} (version 9.0.1) and \cbc{}\footnote{https://github.com/coin-or/Cbc} (version 2.9.0).
More data regarding this experiment can be found in~\Cref{sec:appendix:exp} (\Cref{tab:detailedtable}).

The evaluation shows that for instances which have a favourable directed path decomposition (provided it can be easily computed) it may pay off to use~\Cref{alg:treelike}.
While a result is not returned within $900$ seconds using the MILP-based approaches for the larger threshold and instances with $N \geq 30$, our implementation returns in less than $100$ seconds for instances up to $N = 88$.
Still, even for these instances it has an exponential increase in runtime and doesn't scale to very large state spaces.

\begin{figure}[tbp]
  \begin{subfigure}{.46\linewidth}
    \centering
    \includegraphics[width=\textwidth]{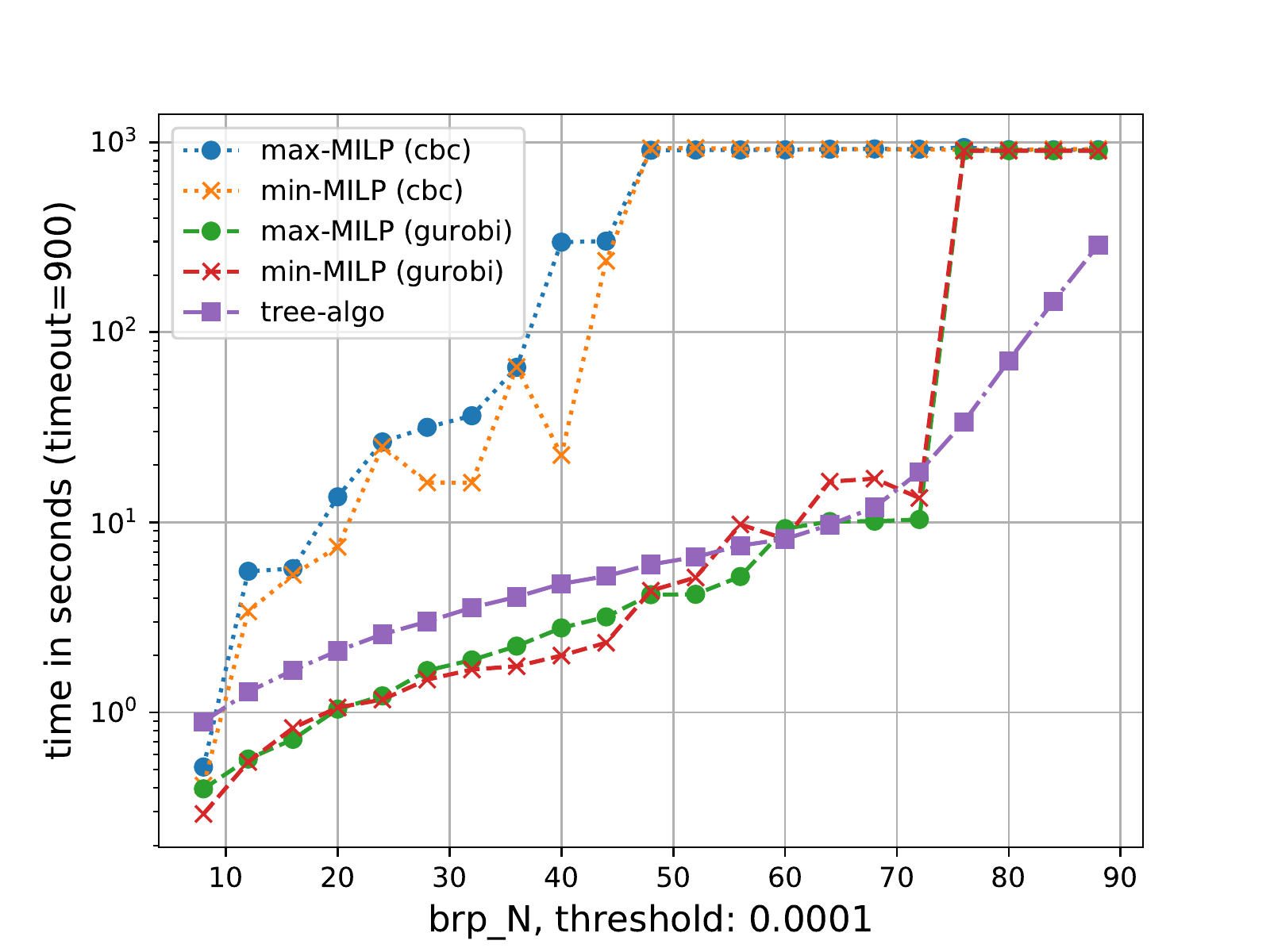}
  \end{subfigure}
  \hfill
  \begin{subfigure}{.46\linewidth}
    \centering
    \includegraphics[width=\textwidth]{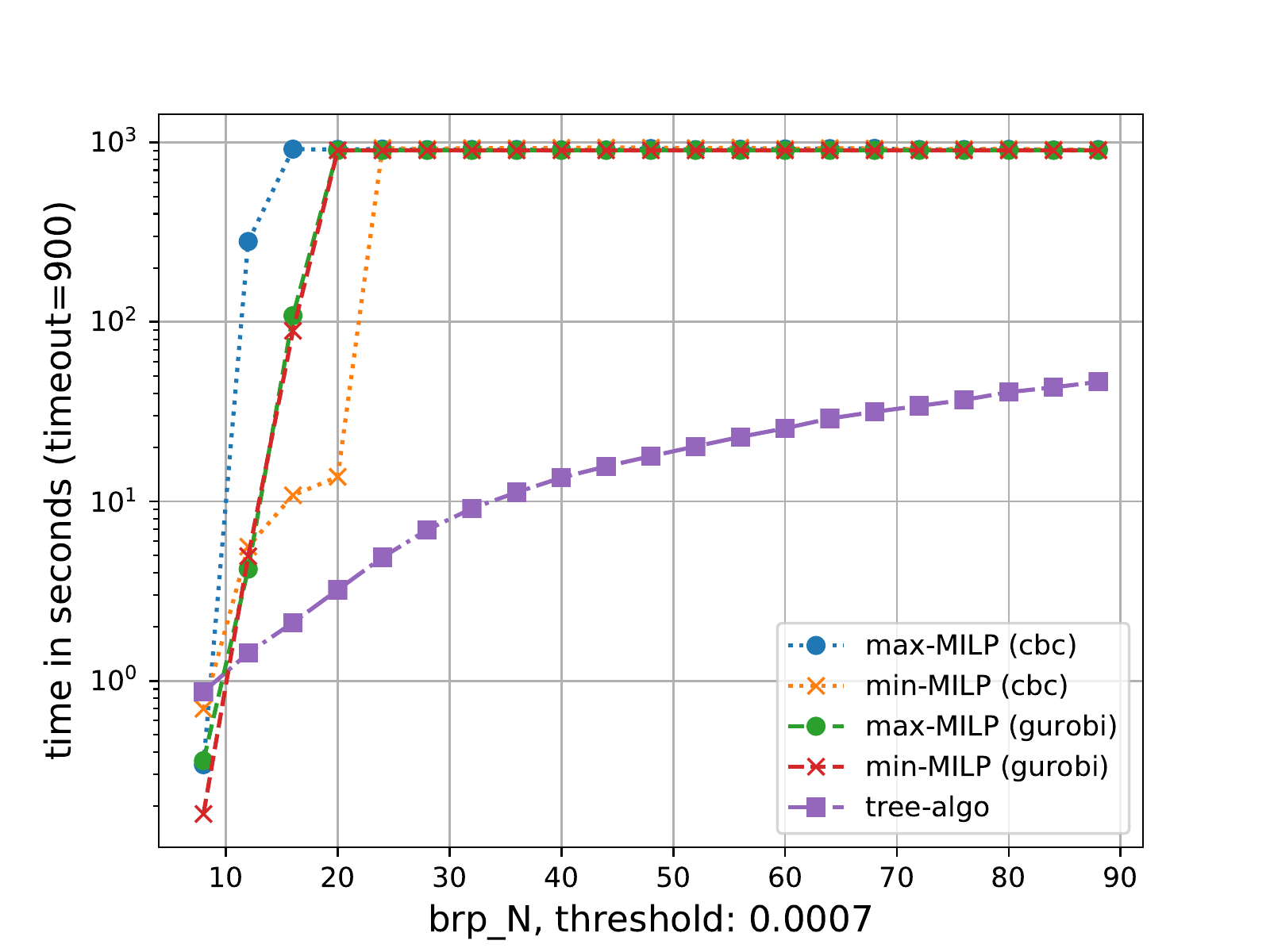}
  \end{subfigure}
  \caption{The computation times of the MILP approaches and~\Cref{alg:treelike} for two different thresholds.}
  \label{fig:experiments}
\end{figure}

\section{Conclusion}

\begin{table}[tbp]
  \caption{Instances of the algorithm.}
  \label{table:instances}
  \renewcommand\arraystretch{1.5}
  \footnotesize
  \centering
  \resizebox{0.9\linewidth}{!}{
  \begin{tabular}{*{2}{>{\raggedright\arraybackslash}p{0.15\linewidth}}*{2}{>{\raggedright\arraybackslash}p{0.3\linewidth}}}
    \toprule
    Model & value function & computing the value (\Cref{alg:treelike:computeval}) & domination relation (\Cref{alg:treelike:removedom}) \\ \hline
    DTMC & $\Pr_{S'}(\lozenge \goal)$ & linear equations & standard \\
    \multirow{2}{*}{MDP} &  $\prb^{\max}_{S'}(\lozenge \goal)$ & \multirow{2}{*}{linear program} & standard \\
    & $\prb^{\min}_{S'}(\lozenge \goal)$ & & strong \\
    \bottomrule
  \end{tabular}
  }
\end{table}

This paper considered the problem of computing minimal witnessing subsystems for probabilistic systems whose underlying graph has low tree width.
The main result is that the corresponding decision problem remains NP-hard for systems with bounded directed tree partition width.
To prove this, the \emph{matrix-pair chain problem} is introduced and shown to be NP-hard for fixed-dimension nonnegative matrices.
In a second step, this problem is reduced to the witness problem.
Finally, an algorithm is described which takes as input a directed tree partition of the system and computes a minimal witnessing subsystem, aiming to utilize the special structure of the system.
A preliminary experimental analysis shows that it outperforms existing approaches for a standard benchmark that allows a good tree partition.

A direction for future work, which would help enabling practical usage of the algorithm described in this paper, is to study how to compute good directed tree partitions, or to characterize systems which allow for natural ones.
Another direction would be to find algorithms which work on standard tree decompositions of the system, as approximation techniques exist to compute them.
Furthermore, it would be interesting to consider heuristic or approximate approaches that utilize the knowledge of a given directed tree partition.
For example, the algorithm described in this paper could be adapted to only store a fixed number of partial subsystems for each block.

\bibliographystyle{splncs04}
\bibliography{references}

\newpage
\appendix

\section{Proofs for~\Cref{sec:treelikesystems}}

\relationToOtherNotions*
\begin{proof}
  Let $G = (V,E)$ be a directed graph and $G_u$ be the undirected graph induced by $G$.
  Let us denote by \utw($G$) the standard notion of (undirected) treewidth of $G_u$, by \dtw($G$) the notion of \emph{directed tree width} from~\cite{JohnsonRST2001}, by \Dw($G$) the notion of D-width from~\cite{Safari2005} and by \utpw($G$) the notion of undirected tree-partition-width of $G_u$ from~\cite{Seese1985}.
  Our aim ist to show that the width of $G$ with respect to any of these notions is bounded from above by a function in \dtpw($G$).

  \textbf{undirected tree (partition) width.} First we observe that any directed tree partition of $G$ directly yields a tree partition of $G_u$ of the same size.
  It follows that $\utpw(G) \leq \dtpw(G)$.
  It was shown in\cite[Fact 2.]{Seese1985} that $2 \: \utpw(G) \geq \utw(G) + 1$.
  Hence the standard undirected tree width \utw($G$) of $G$ is also bounded from above by a function in $\dtpw(G)$.
  
  \textbf{D-width}. A \emph{$d$-decomposition} of $G$ is a pair $(T,X)$ where $T$ is a tree and $X$ is a function which labels the nodes of $T$ by subsets of $V$ such that: 1. all vertices of $G$ appear in one of the sets and 2. for every strongly connected component $\cal S$ of $G$ the nodes $t$ of $T$ such that $X(t) \cap \mathcal{S} \neq \varnothing$ form a connected subtree of $T$~\cite{Safari2005}.
  Clearly, a directed tree partition satisfies this property as every strongly connected component needs to be contained in a single block.
  Hence every directed tree partition induces a $d$-decomposition of the same size, which implies $\Dw(G) \leq \dtpw(G)$.

  \textbf{directed tree width.} It is shown in~\cite[Corollary 1.]{Safari2005} that the directed tree width of any graph is smaller than its $D$-width ($\dtw(G) \leq \Dw(G)$), and hence it follows that $\dtw(G) \leq \dtpw(G)$.
\end{proof}

\treePartitionWithNPHardness*
\begin{proof}
  Membership in NP holds in both cases as one can guess a partition $\cal P$ and check whether it is a valid directed path partition (resp. directed tree partition) and whether it satisfies $\max_{S \in {\cal P}} |S| \leq k$.

  For hardness, we reduce from the \emph{oneway bisection problem} of directed graphs, which was shown to be NP-hard in~\cite{FeigeY2003}.
  It asks, given a directed graph $G$, whether there exists a bisection $V_0,V_1$ of $G$ (that is a partition of the vertices into $V_0$ and $V_1$ satisfying $|V_0| = |V_1|$) such that there are no directed edges from $V_1$ to $V_0$.
  To reduce the oneway bisection problem to the question of whether the directed path width is at most $k$, let us fix a graph $G = (V,E)$.
  Let us construct a new graph $G' = (V \cup \{i,e\}, E')$ (assuming $\{i,e\} \cap V = \varnothing$), where $E' = E \cup \{(i,v), (v,e) \mid v \in V\}$.
   We claim:
  \[\dppw(G') \leq \frac{|V|}{2} + 1 \quad \text{ if and only if } \quad G \text{ has a oneway bisection}\]
  Suppose first that $G$ has a oneway bisection $V_0, V_1$.
  Then $(\{i\} \cup V_0, \{e\} \cup V_1)$ a directed path partition of $G'$.
  This follows directly from the fact that there is no directed edge from $V_1$ to $V_0$.
  The width of this path partition is $|V|/2 + 1$, as $|V_0| = |V_1| = |V|/2$.

  For the other direction, we first observe that any directed path partition of $G'$ has length between one and three.
  This can be seen as follows.
  Vertex $e$ must appear in one of the first three blocks, as any vertex of $G'$ has a path to $e$ of length at most three.
  This also implies that all other vertices must be part of a block preceeding the block that contains $e$.
  
  We now show that a path partition of $G'$ with width at most $|V|/2 + 1$ has length two.
  It cannot have length one, as the single block would then have to contain all vertices.
  So suppose that it has length three.
  Then the first block, which must include $i$, cannot include any other vertex $v \in V$.
  This is because then $e$ must be contained in the first or second block, as there exists an edge from $v$ to $e$.
  In both cases, the third block remains empty.
  At the same time, no vertex $v \in V$ can be included in the third block, as it is reachable from $i$ in a single step.
  Hence the second block contains all $|V|$ vertices, contradicting the fact that the width is at most $|V|/2 + 1$.

  So take a path partition of length two with width at most $|V|/2 + 1$.
  Then, the two blocks have exactly $|V|/2 + 1$ elements, and hence $|V|/2$ vertices from $V$ respectively.
  This partition of $V$ induces a oneway bisection of $G$ as there cannot be any directed edges from the second block to the first one.

  To see that deciding $\dtpw(G) \leq k$ is also NP-hard it suffices to observe that the only directed tree partitions of the graph $G'$ as used in the above reductions are already directed path partitions, as $e$ is reachable from all vertices.
\end{proof}

\section{Proofs for~\Cref{sec:hardnessmcp}}\label{app:proofsMCP}

\twoMCP*

\begin{proof}
To show NP-hardness of the $2$-MCP, we reduce from the \emph{partition} problem, which is among Karp's 21 NP-complete problems\cite{Karp1972}.
Given a finite set $S = \{s_1, \ldots, s_n\} \subseteq \mathbb{Z}$, it asks to decide whether there exists $W \subseteq S$ such that $\sum W = \sum (S \setminus W)$, where $\sum X = \sum_{x \in X} x$.
The main idea of the construction is to relate each entry $s_i \in S$ to a pair of matrices $M^{i}_0,M^{i}_1$ where $M^{i}_0$ is a two-dimensional matrix realizing the clockwise rotation by an angle which corresponds to the 
 value of $s_i$, and $M^{i}_1$ realizes the counter-clockwise rotation by the same angle.
Then for all $\sigma_1, \ldots, \sigma_n \in \{0,1\}^n$ we have that $W = \{s_i \mid \sigma_i = 1\}$ satisfies $\sum W = \sum (S \setminus W)$ iff $M^1_{\sigma_1} \cdots M^n_{\sigma_n}$ equals the identity matrix, which is used in the sequel.

  Let $S = \{s_1, \ldots, s_n\} \subseteq \mathbb{Z}$ be an instance of the partition problem and $m$ be the maximal absolute value that can be accumulated by any subset of $S$, that is $m = \max \{\sum S \cap \mathbb{Z}_{>0}, - \sum S \cap \mathbb{Z}_{<0}\}$.
  We let $\gamma = 3/(4m)<\pi/(4m)$, which is the granularity of rotation we will consider.
  Let $R^{+\varphi}$ be the rotation matrix (in $\mathbb{R}^{2\times 2}$) that rotates a point clockwise by $\varphi$ and $R^{-\varphi}$ be the matrix that rotates counter-clockwise by $\varphi$ (assuming $\varphi \geq 0$ represents an angle in radian).
  Furthermore, let $M_0^{i} = R^{s_i \gamma}$ and $M_1^{i} = R^{- s_i \gamma}$.
  We observe that for all $\sigma_1,\ldots,\sigma_n \in \{0,1\}^n$ we have:
\[M_{\sigma_1}^1 \cdots M_{\sigma_n}^n = I \quad \iff \quad \sum_{\sigma_i = 0} s_i  = \sum_{\sigma_i = 1} s_i \tag{$\ast$}\]
  where $I$ is the identity matrix in two dimensions.
  This uses that $\gamma$ is chosen in a way that prevents a total rotation by more than $\pi/4$.
  Let us fix $\iota = (1/2, 1/2)$, $f = (1/2, 1/2)^T$ and $\lambda = 1/2$.
  We claim that there exist $\sigma_1,\ldots,\sigma_n \in \{0,1\}^n$ such that
  \[\iota \cdot M_{\sigma_1}^1 \cdots M_{\sigma_n}^n \cdot f \geq \lambda\]
  if and only if $S$ is a yes-instance of the partition problem.
  If this is the case, then by ($\ast$) we find $\sigma_1,\ldots,\sigma_n \in \{0,1\}^n$ such that $\iota \cdot M_{\sigma_1}^1 \cdots M_{\sigma_n}^n \cdot f = 1/2$.
  If it is not, then for all $\sigma_1,\ldots,\sigma_n \in \{0,1\}^n$ the product $M_{\sigma_1}^1 \cdots M_{\sigma_n}^n$ is different from $I$ by ($\ast$).
  By construction $M_{\sigma_1}^1 \cdots M_{\sigma_n}^n \cdot f$ is a point on the circle of radius $1/\sqrt{2}$ and by the previous observation it is not $(1/2, 1/2)$.
  It can be easily checked that the sets $\{ v \mid (1/2, 1/2) \cdot v \geq 1/2 \}$ and $\{ v \mid |v| = 1/\sqrt{2}\}$ intersect only in the point $(1/2,1/2)$ and hence $\iota \cdot M_{\sigma_1}^1 \cdots M_{\sigma_n}^n \cdot f < \lambda$ must hold.

  The fact that the constructed matrices may be irrationally-valued dissallows using them directly as input for $2$-MCP.
  It is shown in\cite{CannyDR1992} that for any rational rotation angle $\varphi$ and $\epsilon \in \mathbb{Q}_{>0}$, a rotation matrix $R^{\theta}$ (rotating by angle $\theta$) with rational entries can be computed in polynomial time in $\log (1/ \epsilon)$ such that $|\varphi - \theta| < \epsilon$.
  Let $\epsilon < 3/(8mn)$ and replace all $M_i^j$ by the result $N_i^j$ of the mentioned algorithm.
  Let $\sigma_1,\ldots,\sigma_n \in \{0,1\}^n$ and consider $N_{\sigma_1}^1 \cdots N_{\sigma_n}^n = R^{\varphi_1}$ and $M_{\sigma_1}^1 \cdots M_{\sigma_n}^n = R^{\varphi_2}$.
  We have $|\varphi_1 - \varphi_2| < n \cdot \epsilon < 3/8m = \gamma/2$.
  Finally, we have to adapt the threshold to account for the error terms.
  We choose $\lambda'$ such that
  \[(1/2,1/2) \cdot R^{n \cdot \epsilon} \cdot
  \begin{pmatrix}
    1/2 \\ 1/2
  \end{pmatrix} \geq \lambda' > (1/2,1/2) \cdot R^{\frac{\gamma}{2}} \cdot
  \begin{pmatrix}
    1/2 \\ 1/2
  \end{pmatrix}\]
  Now if $M_{\sigma_1}^1 \cdots M_{\sigma_n}^n = I$ we have $\varphi_2 = 0$ and hence $|\varphi_1| < n \cdot \epsilon$.
  It follows that
  \[(1/2, 1/2) \cdot R^{\varphi_1} \cdot \begin{pmatrix}
    1/2 \\ 1/2
  \end{pmatrix} \geq (1/2,1/2) \cdot R^{n \cdot \epsilon} \cdot
  \begin{pmatrix}
    1/2 \\ 1/2
  \end{pmatrix} \geq \lambda'\]
  If $M_{\sigma_1}^1 \cdots M_{\sigma_n}^n \neq I$ we have $|\varphi_2| \geq \gamma$, as $\gamma$ is granularity of the rotations defined by $M_j^i$.
  In that case $|\varphi_1| \geq \gamma/2$ and hence
  \[(1/2, 1/2) \cdot R^{\varphi_1} \cdot \begin{pmatrix}
    1/2 \\ 1/2
  \end{pmatrix} \leq (1/2,1/2) \cdot R^{\frac{\gamma}{2}} \cdot
  \begin{pmatrix}
    1/2 \\ 1/2
  \end{pmatrix} < \lambda'\]
  It follows that there exists a sequence $\sigma_1,\ldots,\sigma_n \in \{0,1\}^n$ such that
  \[(1/2,1/2) \cdot N_{\sigma_1}^1 \cdots N_{\sigma_n}^n \cdot \begin{pmatrix}
    1/2 \\ 1/2
  \end{pmatrix} \geq \lambda' \]
  if and only if $W = \{s_i \mid \sigma_i = 1\}$ satisfies $\sum W = \sum (S \setminus W)$.
\end{proof}

\threeMCP*

\begin{proof}
  The proof goes by reduction from $2$-MCP.
  Let $(M_0^1,M_1^1),\ldots,(M_0^n,M_1^n)$, with $M_i^j \in \mathbb{Q}^{2 \times 2}$ for all $(i,j) \in \{0,1\} \times \{1,\ldots,n\}$, and $\iota \in \mathbb{Q}^{1 \times 2},f \in \mathbb{Q}^{2 \times 1}$ be an instance of $2$-MCP.
  Define: \[B = 
  \begin{pmatrix}
    1 & 1 & 1 \\
    -1 & 1 & 1 \\
    0 & -2 & 1
  \end{pmatrix}
  \qquad \text{with inverse } \quad B^{-1} = 
  1/6\cdot \begin{pmatrix}
    3 & -3 & 0 \\
    1 & 1 & -2 \\
    2 & 2 & 2
  \end{pmatrix}
  \]

  For a given $\kappa \in \mathbb{Q}$, we define:

  \begin{equation}
    \label{eq:newinstance}
    N_i^j = 
    B \begin{pmatrix}
      M_i^j & \mathbf{0} \\ \mathbf{0} & \kappa
    \end{pmatrix} B^{-1}, \quad
    \iota' = 
    \begin{pmatrix}
      \iota & \kappa
    \end{pmatrix} B^{-1}, \quad
    f' =
    B\begin{pmatrix}
    f \\ \kappa
    \end{pmatrix} \;\; \text{ and } \;\; \lambda' = \lambda + \kappa^{n+2}
  \end{equation}
  For any $\sigma_1,\ldots,\sigma_n \in \{0,1\}^n$ we get:
  \begin{align*}
    N_{\sigma_1}^1 \cdots N_{\sigma_n}^n =
    B \begin{pmatrix}
    M_{\sigma_1}^1 & \mathbf{0} \\ \mathbf{0} & \kappa
  \end{pmatrix} B^{-1} \cdot B \begin{pmatrix}
    M_{\sigma_{2}}^{2} & \mathbf{0} \\ \mathbf{0} & \kappa
  \end{pmatrix} B^{-1} \cdots B \begin{pmatrix}
    M_{\sigma_n}^n & \mathbf{0} \\ \mathbf{0} & \kappa
  \end{pmatrix} B^{-1} = 
    B \begin{pmatrix}
      M_{\sigma_1}^1 \cdots M_{\sigma_n}^n & \mathbf{0} \\ \mathbf{0} & \kappa^n
    \end{pmatrix} B^{-1}
  \end{align*}
  Applying initial and final weights and comparing with the threshold yields:
  \begin{align*}
    &\iota' \cdot B \begin{pmatrix}
      M_{\sigma_n}^n \cdots M_{\sigma_1}^1 & \mathbf{0} \\ \mathbf{0} & \kappa^n
    \end{pmatrix} B^{-1} \cdot f' \geq \lambda' \\
    \iff & \begin{pmatrix}
    \iota & \kappa
  \end{pmatrix} \cdot B^{-1} \cdot B \begin{pmatrix}
      M_{\sigma_n}^n \cdots M_{\sigma_1}^1 & \mathbf{0} \\ \mathbf{0} & \kappa^n
    \end{pmatrix} B^{-1} \cdot B \cdot \begin{pmatrix}
    f \\ \kappa
    \end{pmatrix} \geq \lambda' \\
    \iff & \iota \cdot M_{\sigma_n}^n \cdots M_{\sigma_1}^1 \cdot f + \kappa^{n+2} \geq \lambda + \kappa^{n+2} \\
    \iff & \iota \cdot M_{\sigma_n}^n \cdots M_{\sigma_1}^1 \cdot f \geq \lambda \\
  \end{align*}
  It remains to find $\kappa$ such that all matrices defined in \Cref{eq:newinstance} are nonnegative.
  We observe that $N_i^{j}$ can be written as 
  \begin{equation}
    \label{eqn:specmatrices}
    N_i^{j} = A_i^j + \frac{2 \kappa}{6} \cdot \mathbf{1}^{3 \times 3}
  \end{equation}
  where $\mathbf{1}^{3 \times 3}$ is the three times three matrix containing just ones, and $A_i^j$ is easily computable in polynomial time from $B, B^{-1}$ and $M_i^j$, which can be seen as follows:
  \begin{align*}
    B \begin{pmatrix}
    M_{i}^j & \mathbf{0} \\ \mathbf{0} & \kappa
  \end{pmatrix} B^{-1} \;\; = \;\;
\begin{pmatrix}
    1 & 1 & 1 \\
    -1 & 1 & 1 \\
    0 & -2 & 1
  \end{pmatrix} \cdot 
  \begin{pmatrix}
    a & b & 0 \\
    c & d & 0 \\
    0 & 0 & \kappa
  \end{pmatrix} \cdot
  1/6
  \begin{pmatrix}
    3 & -3 & 0 \\
    1 & 1 & -2 \\
    2 & 2 & 2
  \end{pmatrix} = \\
  1/6 \left(
\underbrace{\begin{pmatrix}
    3(a{+}c) + b + d & - 3(a{+}c) + b + d & - 2(b {+} d) \\
    3(c{-}a) - b + d & - 3(c{-}a) - b + d & - 2(d {-} b) \\
    -6c - 2d & 6c - 2d & 4d
  \end{pmatrix}}_{6\cdot A_i^j}
  + 
  \begin{pmatrix}
    2\kappa & 2\kappa & 2\kappa \\
    2\kappa & 2\kappa & 2\kappa \\
    2\kappa & 2\kappa & 2\kappa \\
  \end{pmatrix} \right)
  \end{align*}
  Hence, it is enough to choose $\kappa$ such that $2 \kappa /6$ is larger than any entry in $A_i^j$ for each $i,j$.
  Similar equations hold for $\iota'$ and $f'$, which means that one can find a $\kappa$ in polynomial time such that all matrices in~\Cref{eq:newinstance} contain only positive entries and concludes the proof.
\end{proof}

\subsection{The MCP for nearly equal-valued matrices.}
\label{sec:matrassumptions}
Finally, we argue that the nonnegative 3-MCP remains hard under the assumption that we have made.
Namely, that all appearing numbers of the 3-MCP instance apart from $\lambda$ are in the range $[1/12-\epsilon,1/12]$ for an $\epsilon$ that satisfies~\Cref{eqn:epsbound}.
Let us first fix a value for $\epsilon$ which depends solely on the length $n$ of an MCP-instance.
We define $\epsilon = 1/ (3 \cdot 12^{n+3} \cdot 2^{n+2})$ and argue that this choice satisfies~\Cref{eqn:epsbound}, which requires:
\[0 < 12 \epsilon < 1/3 \cdot \big(1/12 - \epsilon\big)^{n+2}\]
To see this, observe first that:
\[1/3 (1/12  - \epsilon)^{n+2} = \frac{1}{3\cdot 12^{n+2}} (1 - \frac{1}{3 \cdot 12^{n+2} \cdot 2^{n+2}})^{n+2} > \frac{1}{3 \cdot 12^{n+2}} \cdot (1- \frac{1}{2^{n+2}})^{n+2} \]
Furthermore:
\begin{align*}
  & \frac{1}{3 \cdot 12^{n+2}} \cdot (1- \frac{1}{2^{n+2}})^{n+2} > 12 \epsilon \iff \big(1- \frac{1}{2^{n+2}}\big)^{n+2} > \frac{1}{2^{n+2}} \iff \big(2 - \frac{2}{2^{n+2}}\big)^{n+2} > 1 \iff n \geq 0
\end{align*}

The proof of~\Cref{prop:nonnegativematrixchain} shows that the nonnegative MCP is already NP-hard when restricted to instances of the form:
\[N_i^{j} = A_i^j + \kappa' \cdot \mathbf{1}^{3 \times 3}, \qquad f' = f + \kappa' \cdot \mathbf{1}^3, \qquad \iota' = \iota + \kappa' \cdot \mathbf{1}^3, \qquad \lambda' = \lambda + \kappa^{n+2} \]
with $0 \leq i \leq 1$ and $1 \leq j \leq n$ and $\kappa' = 2 \kappa/6$ for large enough $\kappa$.
This 3-MCP instance can be transformed (in polynomial time in the size of $A_i^j,f,\iota$ and $\lambda$) into an equivalent one which satisfies our assumptions as follows.
Let $max$ and $min$ be the maximal and minimal values appearing in $A_i^j,f,\iota$. Define
\[M_i^j = \frac{1}{12(max + \kappa')} N_i^j, \qquad f_{new} = \frac{1}{12(max + \kappa')} f', \qquad \iota_{new} = \frac{1}{12(max + \kappa')} \iota'\]
The largest entry of $M_i^j,f_{new}$ and $\iota_{new}$ is $1/12$, and the smallest one is $1/12 \cdot (min + \kappa')/(max + \kappa')$.
The largest difference between any two entries is then: $(max-min)/(12(max + \kappa'))$.
Accordingly, we choose $\kappa'$ as follows:
\[\kappa' = \frac{max - min}{12 \epsilon} - max\]
This ensures that all entries of matrices $M_i^j$ and vectors $f_{new}$ and $\iota_{new}$ are in the range $[1/12 - \epsilon, 1/12]$.
As all numbers defining $\kappa'$ have a polynomial binary representation, we can compute it in polynomial time and with it the new 3-MCP instance $M_i^j$ (with $0 \leq i \leq 1$, $1 \leq j \leq n$), $f_{new},\iota_{new}$ and $\lambda_{new} = \lambda'/(12(max + \kappa'))$.
It satisfies:
\begin{lemma}
  For any $\sigma_1,\ldots,\sigma_n \in \{0,1\}^n$ we have:
  \[\iota_{new} \cdot M_{\sigma_1}^1 \cdots M_{\sigma_n}^n \cdot f_{new} \geq \lambda_{new} \; \iff \; \iota' \cdot N_{\sigma_1}^1 \cdots N_{\sigma_n}^n  \cdot f' \geq \lambda'\]
\end{lemma}

\section{Proofs of Section \Cref{sec:hardness_witness}}\label{app:hardness_witness}

\constrWidth*
\begin{proof}
  Partitioning the states of $\M_1$ along the $n+1$ layers yields a directed path partition with partition-width $6$.
  It remains to argue that there is no directed tree partition with a smaller width.
  First we observe that any directed tree partition of $\M_1$ is a path partition.
  This is because for all states pairs of states of $\M_1$ are able to reach the state $\goal$, and hence cannot belong to independent parts of a tree partition.
  Let $\{B_1, \ldots, B_m\}$ be a directed path partition of $\M_1$ where we assume that $B_{i+1}$ is the successor of $B_i$ in the path order, for $1 \leq i \leq m$.
  Now consider arbitrary state of $\M_1$ apart from $\{x_{n+1},y_{n+1},z_{n+1},\goal\}$, for example $\lst{i}{x}$ for any $1 < i < n$.
  Let $B_j$ be the block such that $\lst{i}{x} \in B_j$.
  As $\lst{i}{x}$ is connected by an incoming edge to all states in $\llay{i-1} \cup \rlay{i-1}$ and by an outgoing edge to all states in $\llay{i+1} \cup \rlay{i+1}$, it follows that all these states must be included in the three blocks $B_{j-1},B_{j}, B_{j+1}$ (if $j = 1$ or $j = m$ there are only two blocks and the argument follows in the same way).
  The same holds for the states in $\llay{i} \cup \rlay{i}$, which means that at least $18$ states are included in the three blocks $B_{j-1},B_{j}, B_{j+1}$.
  Then, by the pigeon hole principle, one of the blocks $B_{j-1},B_{j}, B_{j+1}$ includes at least six states.
\end{proof}

\goodSubsysyProp*
\begin{proof}
  Let $\nu(s)$ be the probability of reaching $\goal$ from state $s$ in the DTMC induced by subsystem $S_{\sigma}$ and define for all $1 \leq j \leq n$: $(x_{j},y_j,z_j) = (\lst{j}{x},\lst{j}{y},\lst{j}{z})$ if $\sigma_j = 0$, and else $(x_{j},y_j,z_j) = (\rst{j}{x},\rst{j}{y},\rst{j}{z})$.
  We show by induction on $i$ (for $0 \leq i \leq n$) that
  \[
  \begin{pmatrix}
    \nu(x_{n+1{-}i}) \\
    \nu(y_{n+1{-}i}) \\
    \nu(z_{n+1{-}i})
  \end{pmatrix} = M_{\sigma_{n{+}1{-}i}}^{n{+}1{-}i} \cdots M_{\sigma_n}^n \cdot f \]
  This is enough, as:
  \[Pr_{S_{\boldsymbol{\sigma}}}(\lozenge \goal) = \iota \cdot \begin{pmatrix}
  \nu(x_{1}) \\
  \nu(y_{1}) \\
  \nu(z_{1})
  \end{pmatrix}\]
  For $i = 0$ it is clear, as the probability vector to reach $\goal$ from $(x_{n+1},y_{n+1},z_{n+1})$ is $f$.
  For $i=i'+1$, we have:
  \[\begin{pmatrix}
  \nu(x_{n+1{-}i}) \\
  \nu(y_{n+1{-}i}) \\
  \nu(z_{n+1{-}i})
  \end{pmatrix} = A \cdot \begin{pmatrix}
  \nu(x_{n+1{-}i'}) \\
  \nu(y_{n+1{-}i'}) \\
  \nu(z_{n+1{-}i'})
  \end{pmatrix}\]
  where $A$ contains the pairwise probabilities to reach $(x_{n+1{-}i'},y_{n+1{-}i'},y_{n+1{-}i'})$ from $(x_{n+1{-}i},y_{n+1{-}i},y_{n+1{-}i})$.
  For example, $(A)_{11} = Pr_{S_{\boldsymbol{\sigma}},x_{n+1{-}i}}(\lozenge x_{n+1{-}i'})$.
  But then $A = M^{n{+}1{-}i}_{\sigma_{n+1{-}i}}$ and by induction hypothesis the claim follows.
\end{proof}

\subsection{Singeling out good subsystems.}\label{app:goodsubsysoptimal}
Let $\M_2$ be the Markov chain that we get by substituting all ``matrix gadgets'' in~\Cref{fig:mainstruct} by the construction in~\Cref{fig:matrmult2}.
That is, we add a $\gamma$-cycle to the states $(x_{n+1},y_{n+1},z_{n+1})$ and additionally to each triple of states $\lst{i}{x},\lst{i}{y},\lst{i}{z}$ and $\rst{i}{x},\rst{i}{y},\rst{i}{z}$.
The probabilities on edges between states on layer $i$ and states on layer $i{+}1$ (previously defined using the matrices $M_0^i$ and $M_1^i$ directly) are replaced by the entries of matrix $(M_0^i)'$ or $(M_1^i)'$  (as defined in~\Cref{eqn:gammamatrix}).
To see that the resulting Markov chain really realizes a multiplication by matrices $M_0^i$ and $M_1^i$ in the corresponding layer we now argue that the matrix $M$ in~\Cref{eqn:gammacyclesprob} indeed contains the pairwise reachability probabilities from $(x_i,y_i,z_i)$ to $(x_{i+1},y_{i+1},z_{i+1})$.
Consider first the vector $(v_{xx},v_{yx},v_{zx})$ containing the probabilities to reach $x_{i+1}$ from $x_i$, $y_i$ and $z_i$, respectively.
Let 
\[M' = (1-\gamma) \cdot 
\begin{pmatrix}
  m_{11} & m_{12} & m_{13} \\
  m_{21} & m_{22} & m_{23} \\
  m_{31} & m_{32} & m_{33}
\end{pmatrix}\]
Then we have: \[v_{xx} = (1-\gamma) m_{11} + \gamma \: v_{yx} \qquad v_{yx} = (1-\gamma) m_{21} + \gamma \: v_{zx} \qquad v_{zx} = (1-\gamma) m_{31} + \gamma \: v_{xx}\]
which implies
\[(1-\gamma^3)v_{xx} = (1-\gamma) (m_{11} + \gamma \: m_{21} + \gamma^2 \: m_{31}) \]
and analogously for $v_{yx}$ and $v_{zx}$.
Now one can check that $(v_{xx},v_{yx},v_{zx})^T$ is indeed the first column of the matrix $M$ defined as follows in~\Cref{eqn:gammacyclesprob}:
\[M = 
\frac{1-\gamma}{1-\gamma^3} \cdot
\begin{pmatrix}
  1 & \gamma & \gamma^2 \\
  \gamma^2 & 1 & \gamma \\
  \gamma & \gamma^2 & 1 \\
\end{pmatrix} \
M' \]
The other two columns are derived in completely analogous fashion.

The result of this construction is a valid Markov chain as the entries of all matrices $M'$ are between $0$ and $1/6$ (as shown in~\Cref{sec:interconnectingstates}), and any state has at most $6$ outgoing edges whose probability corresponds to one entry of such a matrix.

By construction of the gadget from~\Cref{fig:matrmult2} it is ensured that if the $\gamma$-cycle \emph{is not interrupted}, then the (pairwise) probabilities of reaching states $(x_{i+1},y_{i+1},z_{i+1})$ from $x_i,y_i,z_i$ are exactly as given by matrix $M_0^i$ (and analogously for $\operatorname{right}$-states and matrices $M_1^i$).
Hence the good subsystems in $\M_2$, which by definition do not break any $\gamma$-cycles, still satisfy the property of~\Cref{lem:goodsybsysprop}.

The fact that all entries of matrices $M_i^j$ (with $1 \leq i \leq n$ and $0 \leq j \leq 1$) and vectors $\iota,f$ have value at least $1/12 - \epsilon$ implies that $\big(3(1/12 - \epsilon)\big)^{n+2}$ is a lower bound on the reachability probability that is achieved by any good subsystem.
This is because assuming that all entries in all matrices and vectors are equal to the lower bound yields a subsystem with this probability, and adding probability to any edge can never decrease the overall probability to reach $\goal$.
Our next aim is to show that no bad subsystem of size $3n +4$ achieves this probability.
\goodsubsystemoptimal*
\begin{proof}
   Let $S_1$ be a \emph{bad} subsystem with at most $3n+4$ states  which include $\goal$.
  By the pigeon hole principle, if there exist a $j$ such that $|\lay{j} \cap \ S_1| > 3$, then there exists $j'$ such that $|\lay{j'} \cap \ S_1| < 3$ (where $1 \leq j,j' \leq n {+} 1$).
  But then both $\gamma$-cycles on layer $\lay{j'}$ are interrupted (or the single one, if $j' = n{+}1$).
  Furthermore, if $S_1$ includes exactly three states in all layers and does not interrupt both $\gamma$-cycles in any layer, then $S_1$ is a good subsystem.

  So it suffices to show that if all $\gamma$-cycles are interrupted in any layer of $S_1$, then the probability to reach $\goal$ is less than $\big(3(1/12 - \epsilon)\big)^{n+2}$, as this is a lower bound on the probability achieved by any good subsystem.
  So let $j$ (with $1 \leq j \leq n{+}1$) be a layer in which both $\gamma$-cycles are interrupted.
  Then, the probability of reaching the next layer (and hence $\goal$) from any state in layer $j$ is at most $3 (1{-}\gamma)$.
  This implies that $Pr_{S_1}(\lozenge \goal)$ is bounded from above by $3(1-\gamma)$.
  Using the assumption that $\gamma$ satisfies~\Cref{eqn:gammachoice} we get:
  \[Pr_{S_1}(\lozenge \goal) \leq \ 3 (1- \gamma) \ < \ \left(3(1/12 - \epsilon)\right)^{n+2}\]
\end{proof}

For completeness we now explain how the the values $f_x',f_y'$ and $f_z'$ in~\Cref{subfig:newmultlastlay} are computed.
The computation is essentially the same as for the main layers.
One can see that the probability to reach $\goal$ from the states $(x_{n+1},y_{n+1},z_{n+1})$ respectively is given by the vector
\[f = 
\underbrace{\frac{1-\gamma}{1-\gamma^3} \cdot
\begin{pmatrix}
  1 & \gamma & \gamma^2 \\
  \gamma^2 & 1 & \gamma \\
  \gamma & \gamma^2 & 1 \\
\end{pmatrix}}_{R} \
f'\]
Solving for $f'$ yields:
\[f' = R^{-1} \cdot f = \frac{1}{1-\gamma}
\begin{pmatrix}
  f(x) - \gamma f(y) \\
  f(y) - \gamma f(z) \\
  f(z) - \gamma f(x)
\end{pmatrix}
\]
This corresponds to the definition of $M'$ in \Cref{eqn:gammamatrix} and it follows in the same way as for the main layers that the entries of $f'$ are between $0$ and $1/6$.

\section{Proofs for~\Cref{sec:algorithm}}

\CompositionLemma*
\begin{proof}
  We will consider the case where $v = \maxval$, the proof for $v = \minval$ is analogous.

  ``$\leq$''. We first prove $\maxval_{S'}(q) \leq \maxval_{S_1'}^f(q)$.
  Let $\S$ be a memoryless scheduler on $\M_{S'}$ such that $\Pr_{\M_{S'}(q)}^{\S}(\lozenge \goal) = \maxval_{S'}(q)$ for all $q \in S'$.
  Let $\S'$ be the scheduler on $\M_{S_1'}^f$ defined exactly as $\S$ on states in $S_1$.
  We claim that $\Pr_{\M_{S_1'}^f(q)}^{\S'}(\lozenge \goal) = \Pr_{\M_{S'}(q)}^{\S}(\lozenge \goal)$ for all $q \in S_1$.
  For states in $\outint(B_i)$ this holds by definition of $f$.
  But, after instantiating the corresponding values of states in $\outint(B_i)$, the standard linear equation system\cite[Theorem 10.19]{BaierK2008} used to compute reachability probabilities for the two Markov chains is the same. 

  ``$\geq$''. We now prove $\maxval_{S'}(q) \geq \maxval_{S_1'}^f(q)$.
  Let $\S$ be a memoryless scheduler on $\M_{S_1'}^f$ such that $\Pr_{\M_{S_1'}^f(q)}^{\S}(\lozenge \goal) = \maxval_{S_1'}^f(q)$ for all $q \in S_1'$.
  Let $\S'$ be a scheduler on $\M_{S'}$ which is defined as $\S$ on states in $S_1$, and for all other states simulates any optimal memoryless scheduler.
  By definition of $f$, $\Pr_{\M_{S_1'}^f(q)}^{\S}(\lozenge \goal) = Pr_{\M_{S'}(q)}^{\S'}(\lozenge \goal)$ holds for states $q \in \outint(B_i)$, and for states $q \in S_1$  this follows as above.
\end{proof}

\begin{lemma}
  \label{lem:valueFuncCompMC}
  Let $\M = (S,P,\iota)$ be a DTMC and $f_1,f_2 : S \to [0,1]$ be two partial functions such that $\dom(f_1) = \dom(f_2)$.
  Let $I \subseteq S$ be a set of states such that no state $q \in I$ can reach $\goal$ in $\M$ without seeing $\dom(f_1)$.
  Then, for all $q \in I$:
  \begin{enumerate}
  \item  If $f_1 + f_2 \leq \mathbf{1}$, then:
    \[\Pr_{\M^{f_1}(q)}(\lozenge \goal) + \Pr_{\M^{f_2}(q)}(\lozenge \goal) = \Pr_{\M^{f_1 + f_2}(q)}(\lozenge \goal).\]
  \item If $a \cdot f_1 \leq \mathbf{1}$, then:
    \[\Pr_{\M^{a \cdot f_1}(q)}(\lozenge \goal) = a \cdot \Pr_{\M^{f_1}(q)}(\lozenge \goal).\]
  \end{enumerate}
\end{lemma}
\begin{proof}
  Let $I^* \supseteq I$ be the closure of $I$ under reachability in $\M^{f_1}$ and $S_? \subseteq I^*$ be the states in $I^*$ which have a path to the set $\dom(f_1)$ in $\M^{f_1}$.
  All states in $I^* \setminus S_?$ have probability $0$ to reach $\goal$ in $\M^f$ for any $f$ with $\dom(f) = \dom(f_1)$, by assumption.
  The vector containing the probabilities to reach $\goal$ from each state in $S_?$ in $\M^{f_1}$ is the unique solution of a linear equation system $(\Ib - \Ab) \xb = \bb$, where $\Ab = (P(s,t))|_{s,t \in S_?}$, $\bb = (P(s,\goal))|_{s \in S_?}$ and $\Ib$ is the identity matrix\cite[Theorem 10.19]{BaierK2008}.
  By construction of $\M^{f_1}$, and our assumption that only states in $\dom(f_1)$ have direct edges to $\goal$, we have $\bb = f_1$, were $f_1$ stands for the vector in $[0,1]^{S_?}$ with entries $0$ for all states outside of $\dom(f_1)$.
  Now if $\xb_1$ is the unique solution of $(\Ib - \Ab) \xb = f_1$ and $\xb_2$ is the unique solution of $(\Ib - \Ab) \xb = f_2$, then $\xb_1 + \xb_2$ is the unique solution of $(\Ib - \Ab) \xb = f_1 + f_2$, and vice versa.
  Similarly, if $a \cdot \xb_1$ is the unique solution of $(\Ib - \Ab) \xb = a \cdot f_1$, then $\xb_1$ is the unique solution of $(\Ib - \Ab) \xb = f_1$, and vice versa.
\end{proof}

\propertiesValuefuncs*
\begin{proof}

  1. Take $f' = f_2 - f_1$. For any scheduler $\S$ on $\M_{T}$ we have:
  \[\Pr^{\S}_{\M^{f_1}_{T}(q)}(\lozenge \goal) + \Pr^{\S}_{\M^{f'}_{T}(q)}(\lozenge \goal) = \Pr^{\S}_{\M^{f_2}_{T}(q)}(\lozenge \goal)\]
  for all states $q \in I$ by~\Cref{lem:valueFuncCompMC}.
  As this holds for all schedulers, it holds in particular for the optimal ones.
  Hence the statement follows for both $\minval$ and $\maxval$.

  2. From~\Cref{lem:valueFuncCompMC} we get that for any scheduler $\S$ on $\M^f$, with $\dom(f) = \dom(f_1)$, and state $q \in I$ we have:
  \[a \cdot \Pr^{\S}_{\M^f(q)}(\lozenge \goal) = \Pr^{\S}_{\M^{a \cdot f}(q)}(\lozenge \goal)\]
  Again, the statement follows for both $\minval$ and $\maxval$ as this holds for all schedulers and in particular the optimal ones.

  3. Assume, for contradiction, that $\maxval_{T}^{f_1 {+} f_2}(q) > \maxval_{T}^{f_1}(q) + \maxval_{T}^{f_2}(q)$ for some $q \in I$.
  Let $\S$ be an optimal scheduler on $\M_{T}^{f_1 + f_2}$.
  By~\Cref{lem:valueFuncCompMC} we have 
  \[\Pr^{\S}_{\M^{f_1 + f_2}_{T}(q)}(\lozenge \goal) = \Pr^{\S}_{\M^{f_1}_{T}(q)}(\lozenge \goal) + \Pr^{\S}_{\M^{f_2}_{T}(q)}(\lozenge \goal)\]
  and, as a consequence:
  \[\maxval_{T}^{f_1}(q) + \maxval_{T}^{f_2}(q) < \Pr^{\S}_{\M^{f_1}_{T}(q)}(\lozenge \goal) + \Pr^{\S}_{\M^{f_2}_{T}(q)}(\lozenge \goal)\]
  But then $\maxval_{T}^{f_i}(q) < \Pr^{\S}_{\M^{f_i}_{T}(q)}(\lozenge \goal)$ holds for some $i \in \{1,2\}$, which is impossible.
  \end{proof}

\begin{remark}
  Point (3) of~\Cref{lem:propertiesValuefuncs} does not holds for $\minval$. Rather, the following holds 
  \[\minval^{f_1 + f_2}_T \geq \minval^{f_1}_T + \minval^{f_2}_{T}\]
  by a similar argument as above for $\maxval$.
  \Cref{fig:minvalnotsublinear} gives an example where the inequation with $\leq$ fails.
\end{remark}
\begin{figure}
  \centering
  \includegraphics{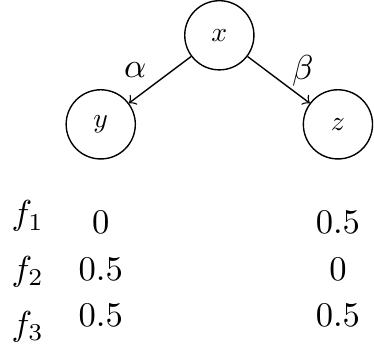}
  \caption{An example showing that $\minval$ does not satisfy (3) of~\Cref{lem:propertiesValuefuncs}. Consider functions $f_1,f_2,f_3 : \{y,z\} \to [0,1]$ defined by $f_1(y) = 0, f_1(z) = 0.5$, $f_2(y) = 0.5, f_2(z) = 0$ and $f_3(y) = 0.5, f_3(z) = 0.5$.
    We have $\minval_{f_3}(x) = \minval_{f_1 + f_2}(x) = 0.5 > \minval_{f_1}(x) + \minval_{f_2}(x) = 0$.}
  \label{fig:minvalnotsublinear}
\end{figure}

\dominationRel*
\begin{proof}
  1. As $\mathcal{S}$ strongly dominates $S_2$ there exists a partial subsystem $T \in \mathcal{S}$ such that $\{T\}$ dominates $S_2$.
  Hence we have $|T| \leq |S_2|$ and $\pointval_I(S_2)$ is a convex combination of $\pointval_I(T)$, which implies that $\pointval_I(S_2) \leq \pointval_I(T)$.
  Let $f_1  = \pointval_I(T)$, $f_2  = \pointval_I(S_2)$ and $S_1' = S_1 \cup I$.
  Then, we have $\minval_{S_1 \cup T}(q) = \minval^{f_1}_{S_1'}(q)$ and $\minval_{S_1 \cup S_2}(q) = \minval^{f_2}_{S_1'}(q)$ for all $q \in S_1'$ by~\Cref{lem:composition}.
  Furthermore, using $f_1 \geq f_2 $ and~\Cref{lem:propertiesValuefuncs} we have $\minval_{S_1 \cup T}(q) \geq \minval_{S_1 \cup S_2}(q)$ for all $q \in S_1'$.
  As all initial states are in the root of the tree partition by assumption, and hence in $S_1'$, it follows that $S_1 \cup T$ induces a witnessing subsystem for $\prb^{\min}(\lozenge \goal) \geq \lambda$.

  2. As $\mathcal{S}$ dominates $S_2$ there exists a subset $\{P_1, \ldots, P_k\} \subseteq \mathcal{S}$ such that $\{P_1, \ldots, P_k\}$ dominates $S_2$.
  Let $f_i = \pointval_I(P_i)$ for $1 \leq i \leq k$ and $g = \pointval_I(S_2)$.
  By definintion, $g$ is a convex combination of vectors $\bigcup \{ \pi(f_i) \mid 1 \leq i \leq k\}$.
  That is, there exists $\lambda_1, \ldots, \lambda_m$ such that:
  \[g = \sum_{j = 1}^m \lambda_j \cdot \gamma_j \qquad \text{ with } \quad \sum_{j = 1}^{m} \lambda_j \leq 1 \quad \text{ and } \quad \gamma_j \in  \bigcup \{ \pi(f_i) \mid 1 \leq i \leq k\} \text{ for } 1 \leq j \leq m\]
  Consider any sequence $\gamma_1', \ldots, \gamma_m' \in \{f_1,\ldots,f_k\}^m$ such that if $\gamma_j' = f_i$, then $\gamma_j \in \pi(f_i)$ for all $1 \leq j \leq m$.
  As $f_i$ is pointwise larger than any vector in $\pi(f_i)$ (for $1 \leq i \leq k$) we have:
  \[g \leq \sum_{j = 1}^m \lambda_j \cdot \gamma_j'\]
  By statements (2.) and (3.) of~\Cref{lem:propertiesValuefuncs} it follows that:
  \[\maxval_{S_1'}^{g} \leq \sum_{j = 1}^m \lambda_j \cdot \maxval_{S_1'}^{\gamma_j'}\]
  Now $\prb^{\max}_{S_1 \cup S_2}(\lozenge \goal) = \iota \cdot \maxval_{S_1'}^g$, and $\prb^{\max}_{S_1 \cup P_j}(\lozenge \goal) = \iota \cdot \maxval_{S_1'}^{f_i}$ by~\Cref{lem:composition} as $\supp(\iota)$ is included in the root of the tree partition, and hence in $S_1$.
  Choose $\gamma_j'$ (with $1 \leq j \leq k$) such that $\iota \cdot \maxval_{S_1'}^{\gamma_j'}$ is maximal and let $f_* = \gamma_j'$.
  Then:
  \[\iota \cdot \maxval_{S_1'}^{f_*} \geq \iota \cdot \sum_{j = 1}^m \lambda_j \cdot \maxval_{S_1'}^{\gamma_j'} \geq \iota \cdot \maxval_{S_1'}^{g}\]
  Hence it follows that $S_1 \cup P_*$ is a witnessing subsystem for $\prb^{\max}(\lozenge \goal) \geq \lambda$.
\end{proof}

\removeDominated*
\begin{proof}
  We adress the first claim.
  For each $k$ in $\{1,\ldots, \max\{|S'| \mid S' \in \mathcal{S}\} \}$ the set $\mathcal{H}$.vertices in~\Cref{alg:convhull:updateR} contains the vertices of the convex hull of $\bigcup\{\pi(\pointval_I(S')) \mid S' \in \mathcal{S}, |S'| \leq k\}$.
  If $\pointval_I(T)$, for $T \in \mathcal{S}[k]$, is not in $\mathcal{H}$.vertices at that point it is a convex combination of $\bigcup\{\pi(\pointval_I(S')) \mid S' \in \mathcal{S}, |S'| \leq k\}$.
  Hence, $T$ is dominated by $\bigcup\{S' \in \mathcal{S}[k] \mid \pointval_I(S') \in \mathcal{H}$.vertices$\}$, and thereby by $R$.
  
  Now suppose that some partial subsystem $T \in R$ is dominated by $R \setminus \{T\}$.
  Then, in particular $T$ is dominated by $\{S' \in R \setminus \{T\} \mid |S'| \leq |T|\}$, and hence also by $\mathcal{S}[|T|]$, as the former is a subset of the latter.
  It follows that $T$ is not a vertex of the convex hull of $\bigcup\{\pi(\pointval_I(S')) \in \mathcal{S}[|T|]\}$, as it is a convex combination of vectors therein.
  But then $T$ cannot be in $R$, as it is not added in~\Cref{alg:convhull:updateR} in the loop iteration corresponding to $k = |T|$.
\end{proof}

\algCorrectness*
\begin{proof}
  First, we argue that if~\Cref{alg:treelike} returns $S'$ then $\M_{S'}$ is a witnessing subsystem for $\prb^{\max}_{\M}(\lozenge \goal) \geq \lambda$.
  This holds because by~\Cref{lem:composition} the values of states in partial subsystems are computed correctly in~\Cref{alg:treelike:computeval}.

  Next, we show that for any witnessing subsystem $\M_{T}$ for the property we have $|T| \geq |S'|$.
  So let $T$ be a witnessing subsystem.
  Let $B_1,\ldots,B_n$ be a reverse-topological order of the tree partition.
  We will construct a sequence $S_0, \ldots, S_n$ of subsets of $S$ inductively such that
  \begin{itemize}
  \item for all $0 \leq i \leq n$: $|S_i| \leq |T|$ and $S_i$ induces a witnessing subsystem for $\prb^{\max}_{\M}(\lozenge \goal) \geq \lambda$.
  \item for all $1 \leq i \leq n$ and $j \leq i$ the partial subsystem $S_i \cap \reach(\interface(B_j))$ for $\interface(B_j)$ is in $\partsubsysmap[B_j]$ at the end of the execution of~\Cref{alg:treelike}.
  \end{itemize}
 
  We start by setting $S_0 = T$.
  To find $S_{i+1}$ we assume that the above properties hold for all $S_j$ with $j \leq i$.
  In case that the partial subsystem $S_i \cap \reach(\interface(B_{i+1}))$ is included in $\partsubsysmap[B_{i+1}]$, we can set $S_{i+1} = S_i$.
  Otherwise, we proceed as follows.
  Let $\{B_{l_1}, \ldots, B_{l_m}\} = \sucs(B_{i+1})$.
  It follows that for each $B_{l_j} \in \sucs(B_{i+1})$ the partial subsystem $P_{l_j} = S_i \cap \reach(\interface(B_{l_j}))$ is in $\partsubsysmap[B_{l_j}]$ at the end of~\Cref{alg:treelike}.
  Hence, the partial subsystem $P = \bigcup\{P_{l_1},\ldots, P_{l_m}\}$ appears in $successorPoints(\partsubsysmap,B_i)$ when considering block $B_i$ in~\Cref{alg:treelike:succpoints} in~\Cref{alg:treelike}.

  We make a case-distinction on whether $S_i \cap B_i$ is a model of the formula $\phi(B_i)$.

  \noindent \textbf{Case 1: $S_i \cap B_i \models \phi(B_i)$}.
  In this case, the partial subsystem $P \cup (S_i \cap B_i)$ is inserted into $\partsubsysmap[B_i]$ in~\Cref{alg:treelike:insertp}.
  As it is not in $\partsubsysmap[B_i]$ at the end of the execution of~\Cref{alg:treelike} by assumption, it must have been removed in~\Cref{alg:treelike:removedom}.
  Hence, by~\Cref{lem:removedom}, $P \cup (S_i \cap B_i)$ is dominated by $\partsubsysmap[B_i]$.
  By~\Cref{lem:dominationrel} we can conclude that that there exists a partial subsystem $P' \in \partsubsysmap[B_i]$ such that $(S_i \setminus \reach(\interface(B_i))) \cup P'$ induces a witnessing subsystem for $\prb^{\max}_{\M}(\lozenge \goal) \geq \lambda$, and $|P'| \leq |P \cup (S_i \cap B_i)|$.
  We set $S_{i+1} = (S_i \setminus \reach(\interface(B_i))) \cup P'$.

  \noindent \textbf{Case 2: $S_i \cap B_i \not\models \phi(B_i)$}.
  In this case there exists some $L \subseteq S_i \cap B_i$ such that $L \models \phi(B_i)$ and for any partial subsystem $P$ for $\outint(B_i)$ the partial subsystem $L \cup P$ for $\interface(B_i)$ strongly dominates the partial subsystem $(S_i \cap B_i) \cup P$.
  Now the argument of \textbf{Case 1} can be applied by observing that if a set of partial subsystems dominate $L \cup P$, then the same set dominates $(S_i \cap B_i) \cup P$.

  This shows that we can construct the sequence $S_0,\ldots,S_n$ satisfying the above properties.
  But then $S_n$ induces a witnessing subsystem and satisfies $|S_n| \leq |T|$.
  As $S_n$ is part of $\partsubsysmap[B_r]$ in the last line of the algorithm, we have $|S'| \leq |S_n|$.

  Finally, we argue that the algorithm takes at most exponential time to return.
  Let $S$ be the states of $\M$, $N = |S|$ and $K$ be the width of the given tree partition.
  The outermost for-loop is taken at most $N$ times.
  The for-loop starting in~\Cref{alg:treelike:BDD} is taken at most $2^K$ times, as it ranges over subsets of $B$ which has at most $K$ states.
  The innermost for-loop in~\Cref{alg:treelike:succpoints} is taken at most $2^N$ times, as it ranges over subsets of $S$.
  The value computation in~\Cref{alg:treelike:computeval} can be done in polynomial time in $\M$.
  The subroutine removeDominated (\Cref{alg:convhull}) which is called in~\Cref{alg:treelike:removedom} requires exponential time in $K$ and polynomial in the size of the input (in this case $\partsubsysmap$[B]).
  This is because the convex hull of $a$ points in dimension $d$ can be computed in time $O(a \log a + a^{\lfloor d/2 \rfloor})$~\cite{Chazelle1993}.
  In our case $d = O(K)$ and $a = O(2^K \cdot |\partsubsysmap[B]|) = O(2^K \cdot 2^N)$.
  The factor of $2^K$ in the number of points used in the convex hull computation comes from the fact that we include all projections of any point for a partial subsystem in $\partsubsysmap[B]$.
  All in all the algorithm requires at most exponential time in the size of $\M$.
\end{proof}

\section{Additional experimental data}
\label{sec:appendix:exp}

\begin{table}[h!]
  \caption{Computation times for computing minimal witnessing subsystems in the bounded retransmission protocol using either the MILP approaches (with solvers cbc and gurobi) or~\Cref{alg:treelike}. The ``min'' and ``max'' in brackets refer to which MILP formulation was used (either the one derived from $\P^{\min}$ or $\P^{\max}$ as defined in~\cite[Lemma 5.1]{FunkeJB2020}).}
  \label{tab:detailedtable}
  \centering
  \includegraphics[width=\textwidth]{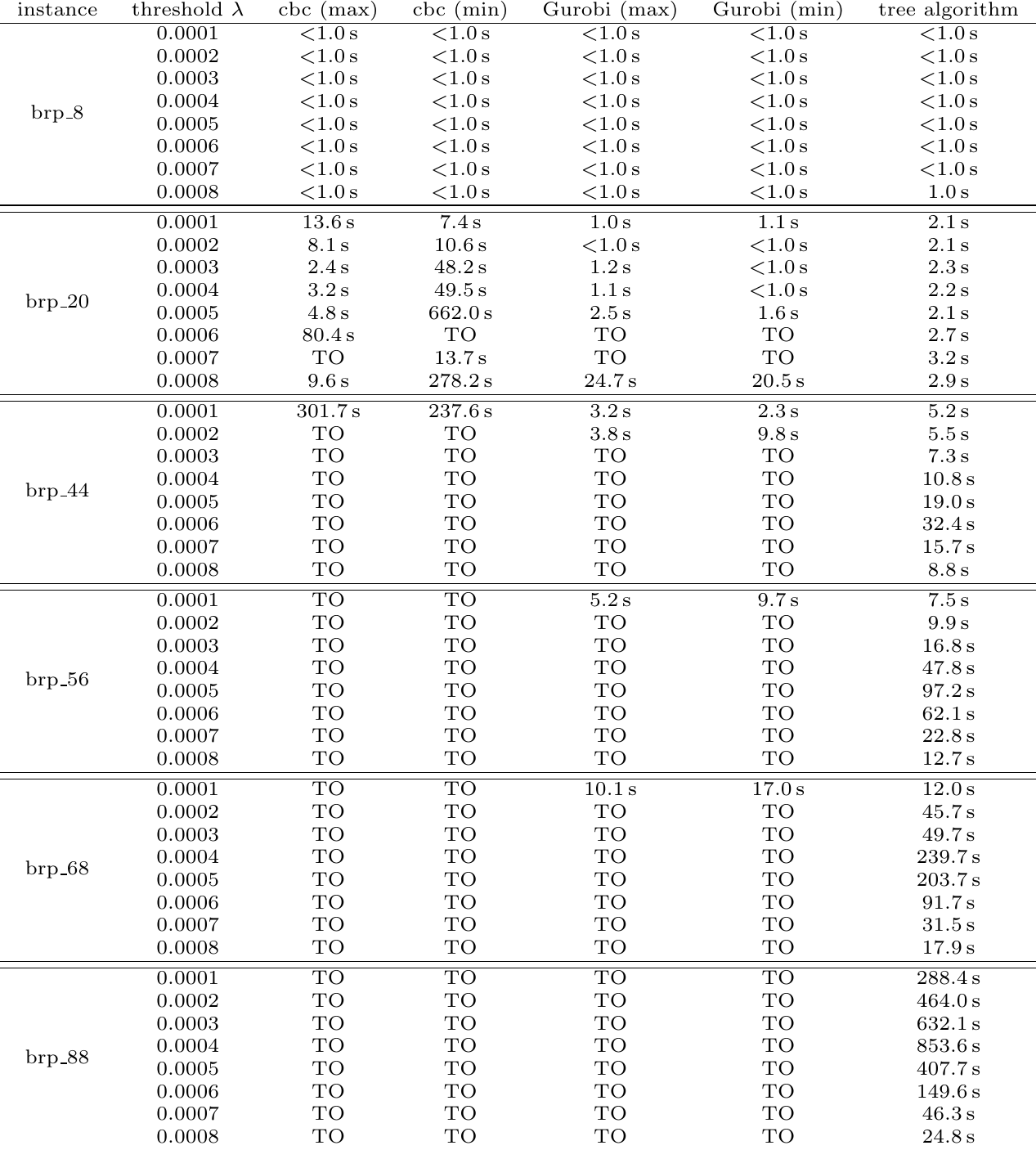}
\end{table}
\end{document}